\documentclass[english]{IEEEtran}

\IEEEoverridecommandlockouts  

\usepackage{fix-cm}
\usepackage{etex}

\usepackage{dblfloatfix}

\usepackage{nag}

\makeatletter
\@ifpackageloaded{xcolor}{}{%
\usepackage[table,x11names,dvipsnames,svgnames]{xcolor}%
}
\makeatother

\usepackage{colortbl}

\usepackage{graphicx}
\usepackage{wrapfig}

\definecolorset{RGB}{lyft}{}{Red,194,39,36;Sunset,202,53,33;Orange,205,68,20;Amber,200,117,42;Yellow,242,169,52;Citron,186,188,44;Lime,112,159,33;Green,56,139,31;Mint,45,118,56;Teal,52,133,135;Cyan,60,132,202;Blue,55,94,248;Indigo,64,13,247;Purple,115,42,248;Pink,176,25,145;Rose,176,32,75}

\usepackage{cite}

\usepackage{microtype}

\usepackage{array}
\usepackage{multirow}
\usepackage{booktabs}
\usepackage{makecell} 

\ifcsname labelindent\endcsname
\let\labelindent\relax
\fi
\usepackage[inline]{enumitem}

\usepackage{subfig}

\newenvironment{lenumerate}[2][]
{\begin{enumerate}[label=(#2\arabic*),leftmargin=0.2in,itemindent=0.15in,#1]}
{\end{enumerate}}

\setlist*[enumerate,1]{label={\itshape\arabic*)}}

\makeatletter
\newcommand{\paragraphswithstop}{%
\let\copyparagraph\paragraph%
\renewcommand\paragraph[1]{\copyparagraph{##1.}}%
}
\makeatother

\usepackage[framemethod=tikz]{mdframed}

\usepackage{suffix}

\usepackage{environ}

\makeatletter
\newsavebox{\boxifnotempty}
\newcommand{\displayifnotempty}[3]{\sbox\boxifnotempty{#2}\setbox0=\hbox{\usebox{\boxifnotempty}\unskip}%
\ifdim\wd0=0pt
\else
 #1\usebox{\boxifnotempty}#3%
\fi%
}
\newcommand{\ifempty}[2]{\setbox0=\hbox{#1\unskip}%
\ifdim\wd0=0pt%
 #2%
\fi%
}
\newcommand{\ifnotempty}[2]{\setbox0=\hbox{#1\unskip}%
\ifdim\wd0>0pt%
 #2%
\fi%
}
\makeatother

\usepackage{algpseudocode}
\usepackage{algorithm}

\usepackage{scrlfile}

\makeatletter
\newcommand*\newstoreddef[1]{
  \BeforeClosingMainAux{%
    \immediate\write\@auxout{%
      \string\restoredef{#1}{\csname #1\endcsname}%
    }%
  }%
}
\newcommand*{\restoredef}[2]{
  \expandafter\gdef\csname stored@#1\endcsname{#2}%
}
\newcommand*{\storeddef}[1]{
  \@ifundefined{stored@#1}{0}{\csname stored@#1\endcsname}%
}
\makeatother

\usepackage{pageslts}
\NewEnviron{tee}{\BODY\typeout{Marker Tee [start] ^^J \BODY ^^JMaker Tee [end]}}

\input{preamble/math}
\newcommand{\real}[1]{\mathbb{R}^{#1}{}}

\newcommand{\naturals}[1]{\mathbb{N}^{#1}{}}

\newcommand{\bmat}[1]{\begin{bmatrix}#1\end{bmatrix}}

\newcommand{\vct}[1]{\mathbf{#1}}

\DeclareMathOperator*{\argmax}{\arg\!\max}

\providecommand{\vw}{\vct{w}}

\providecommand{\cC}{\mathcal{C}}

\providecommand{\cI}{\mathcal{I}}

\providecommand{\cS}{\mathcal{S}}

\providecommand{\cX}{\mathcal{X}}

\usepackage{units}

\newcommand{\newcolorlabel}[2]{%
  \expandafter\newcommand\csname #1\endcsname[1]{%
    \colorbox{#2}{\color{white}\textsf{\textbf{##1}}}}%
}

\newcommand{\newcommenter}[2]{%
  \expandafter\newcommand\csname #1\endcsname[1]{%
    \fcolorbox{#2}{#2}{\color{white}\textsf{\textbf{#1}}}
    {\color{#2}##1}}%
  \expandafter\newcommand\csname at#1\endcsname{%
    \fcolorbox{#2}{#2}{\color{white}\textsf{\textbf{@#1}}}
    {\color{#2}}}%
  \expandafter\newcommand\csname #1hl\endcsname[2]{%
    \colorbox{#2}{\color{white}\textsf{\textbf{#1}}}\sethlcolor{Azure2}\hl{##2}~%
    \expandafter\ifx\csname commentarrow\endcsname\relax$\leftarrow$\else \commentarrow[#2]\fi~%
    {\color{#2}##1}}%
  \expandafter\newcommand\csname #1st\endcsname[2]{%
    \colorbox{#2}{\color{white}\textsf{\textbf{#1}}}\sout{##2}~%
    \expandafter\ifx\csname commentarrow\endcsname\relax$\leftarrow$\else \commentarrow[#2]\fi~%
    {\color{#2}##1}}%
}
\newcommenter{TODO}{DodgerBlue1}
\newcommenter{rtron}{Green3}

\usepackage{comment}

\usepackage{pdfcomment}

\usepackage{soul}

\usepackage[normalem]{ulem}

\usepackage{csquotes}

\usepackage{tikz}
\usetikzlibrary{calc}
\usetikzlibrary{matrix}
\usetikzlibrary{chains}
\usetikzlibrary{shapes.geometric}
\usetikzlibrary{arrows.meta}
\usetikzlibrary{decorations.pathreplacing}
\usetikzlibrary{backgrounds}

\tikzset{
  dim above/.style={to path={\pgfextra{
        \pgfinterruptpath
        \draw[>=latex,|->|] let
        \p1=($(\tikztostart)!1.5em!90:(\tikztotarget)$),
        \p2=($(\tikztotarget)!1.5em!-90:(\tikztostart)$)
        in(\p1) -- (\p2) node[pos=.5,sloped,above]{#1};
        \endpgfinterruptpath
      }
    }
  },
  dim double above/.style={to path={\pgfextra{
        \pgfinterruptpath
        \draw[>=latex,|->|] let
        \p1=($(\tikztostart)!3em!90:(\tikztotarget)$),
        \p2=($(\tikztotarget)!3em!-90:(\tikztostart)$)
        in(\p1) -- (\p2) node[pos=.5,sloped,above]{#1};
        \endpgfinterruptpath
      }
    }
  },
  dim below/.style={to path={\pgfextra{
        \pgfinterruptpath
        \draw[>=latex,|->|] let 
        \p1=($(\tikztostart)!-1em!-90:(\tikztotarget)$),
        \p2=($(\tikztotarget)!-1em!90:(\tikztostart)$)
        in (\p1) -- (\p2) node[pos=.5,sloped,below]{#1};
        \endpgfinterruptpath
      }
    }
  },
}

\tikzset{
    right angle quadrant/.code={
        \pgfmathsetmacro\quadranta{{1,1,-1,-1}[#1-1]}     
        \pgfmathsetmacro\quadrantb{{1,-1,-1,1}[#1-1]}},
    right angle quadrant=1, 
    right angle length/.code={\def\rightanglelength{#1}},   
    right angle length=2ex, 
    right angle symbol/.style n args={3}{
        insert path={
            let \p0 = ($(#1)!(#3)!(#2)$) in     
                let \p1 = ($(\p0)!\quadranta*\rightanglelength!(#3)$), 
                \p2 = ($(\p0)!\quadrantb*\rightanglelength!(#2)$) in 
                let \p3 = ($(\p1)+(\p2)-(\p0)$) in  
            (\p1) -- (\p3) -- (\p2)
        }
    }
}

\newcommand{\pgfextractangle}[3]{%
    \pgfmathanglebetweenpoints{\pgfpointanchor{#2}{center}}
                              {\pgfpointanchor{#3}{center}}
    \global\let#1\pgfmathresult  
}

\usetikzlibrary{shapes.arrows}
\newcommand{\commentarrow}[1][Azure4]{\tikz[baseline=-3pt]{\node[shape border uses incircle, fill=#1,rotate=180,single arrow, inner sep=1pt, minimum size=6pt, single arrow head extend=2pt]{};}}

\tikzset{ax/.style={-latex,line width=2pt}}

\tikzset{camera/.style={fill=Sienna1,fill opacity=0.5},%
image plane/.style={draw=RoyalBlue3,line width=2pt}}

\newcommenter{danyang}{cyan}
\DeclareMathOperator{\E}{\mathbb{E}}
\DeclareMathOperator{\bernoulli}{Ber}

\newtheorem{example}{Example}

\usepackage{cleveref}

\title{\LARGE \bf

TLINet: Differentiable  Neural Network  Temporal Logic Inference

}

\author{Danyang Li$^{1}$, Mingyu Cai$^{3}$, Cristian-Ioan Vasile$^{2}$, Roberto Tron$^{1}$
\thanks{$^{1}$Danyang Li and Roberto Tron are with Mechanical Engineering Department, Boston University, Boston, MA 02215, USA.
        {\tt\small danyangl@bu.edu, tron@bu.edu}}%
\thanks{$^{2}$ C.I. Vasile is with the Department of Mechanical Engineering and Mechanics, Lehigh University, Bethlehem, PA, USA.
        {\tt\small cvasile@lehigh.edu}}%

\thanks{$^{2}$M. Cai is with the Department of Mechanical Engineering and Mechanics, University of California Riverside, Riverside, CA, USA.
{\tt\small mingyu.cai@ucr.edu}}%
}

\begin{document}

\def\next{\circ}
\def\always{\square}
\def\event{\lozenge} 
\def\relu{\text{ReLU}}

\maketitle
\thispagestyle{empty}
\pagestyle{empty}

\begin{abstract}
There has been a growing interest in extracting formal descriptions of the system behaviors from data. Signal Temporal Logic (STL) is an expressive formal language used to describe spatial-temporal properties with interpretability. 
This paper introduces TLINet, a neural-symbolic framework for learning STL formulas. The computation in TLINet is differentiable, enabling the usage of off-the-shelf gradient-based tools during the learning process.
In contrast to existing approaches, we introduce approximation methods for max operator designed specifically for temporal logic-based gradient techniques, ensuring the correctness of STL satisfaction evaluation. Our framework not only learns the structure but also the parameters of STL formulas, allowing flexible combinations of operators and various logical structures.
We validate TLINet against state-of-the-art baselines, demonstrating that our approach outperforms these baselines in terms of interpretability, compactness, rich expressibility, and computational efficiency.
\end{abstract}


\section{INTRODUCTION}
Machine learning techniques, particularly neural networks, have seen widespread application across various fields, including motion planning and control for autonomous systems. Despite the considerable success achieved with neural networks in this domain, their lack of interpretability poses significant challenges. Interpretability refers to the ability to provide explanations in understandable terms to humans~\cite{doshi2017towards}. This limitation makes neural networks difficult to verify, sparking a growing interest in more interpretable models applicable across various tasks, such as debugging and validating AI-integrated systems~\cite{tjoa2020survey}.

One domain where interpretability is particularly crucial is the analysis and interpretation of time series data inherent in dynamical systems like autonomous drones and robot arms. Traditional approaches to applying neural networks to time series data often involve the use of black-box models, such as Recurrent Neural Networks (RNNs)~\cite{schuster1997bidirectional}, Long-Short-Term Memory (LSTM) networks~\cite{hochreiter1997long}, and Transformers~\cite{vaswani2017attention}. Understanding and validating the decisions made by these models remain elusive, raising concerns about their suitability for safety-critical applications.

In response to this challenge, researchers have explored alternative methodologies that offer greater transparency and comprehensibility. One promising approach that has garnered attention is Signal Temporal Logic (STL).
Defined over continuous-time domains, STL provides a formal language for expressing complex temporal and logical properties in a manner resembling natural language~\cite{maler2004monitoring}.
By quantifying the satisfaction of STL specifications, it facilitates various optimization objectives and constraints, offering a principled framework for controlling and reasoning about autonomous systems, spanning applications such as control~\cite{raman2014model, lindemann2018control,liu2024interpretable} and motion planning~\cite{vasile2017sampling}.
Furthermore, STL holds the potential to bridge the gap between raw time series data and interpretable logic specifications, thus addressing the need for formalized, human-understandable representations in machine learning applications.

In this paper, we aim to explore the role of temporal logic inference in extracting interpretable models from time series data.
Specifically, we propose a neural network for learning STL formulas from time series data to classify desired and undesired behaviors as satisfying and violating behaviors, respectively. By designing neural networks for temporal logic inference, we seek to bridge the gap between data-driven insights and formalized logical representations, ultimately paving the way for safe and transparent autonomous systems.

\vspace{0.1cm}
\textbf{Related Works: }
There exists a substantial body of literature on methods for temporal logic inference. Decision tree approach has been explored for deriving temporal logic formulas as classifiers~\cite{bartocci2014data, bombara2016decision, mohammadinejad2020interpretable, bombara2021offline, xu2019information, aasi2022classification, linard2022inference}.
A decision tree resembles a non-parametric supervised learning method where intermediate nodes partition the data based on specific criteria, guiding the flow towards leaf nodes representing the final decision.
While decision trees offer a structured approach to classification, they do not scale well by dataset size and tree depth~\cite{dietterich1997machine, sani2018computational}.
In contrast, neural networks offer scalability through batching or vectorizing the data and utilizing state-of-the-art gradient-based optimization techniques for training.

Many studies have investigated embedding the structure of STL formulas in neural network computation graphs by associating layers with Boolean and temporal operators via smooth approximations of min and max functions~\cite{jha2019telex, leung2019backpropagation, ketenci2020learning, yan2021stone, baharisangari2021weighted, chen2022interpretable}. These studies broadly fall into two categories: template-based learning and template-free learning. Template-based learning involves fixing the structure of the STL formula and only learning its parameters \cite{jha2019telex, leung2019backpropagation, ketenci2020learning, yan2021stone, baharisangari2021weighted}, whereas template-free learning \cite{chen2022interpretable} learns the STL formula without specifying prior structures. Our approach aligns with the template-free learning category. There are several challenges when integrating STL into neural networks.
For example, the recursive min/max operations may lead to gradient vanishing problems; the selection of partial signals from the time intervals associated with the temporal operators, ``always" and ``eventually," is non-differentiable.

To address non-smoothness challenges in backpropagation of neural networks, several studies have proposed differentiable versions of robustness computation for STL~\cite{mehdipour2020specifying,varnai2020robustness,leung2019backpropagation, yan2021stone, chen2022interpretable}. Some investigations, such as those by Yan et al. \cite{yan2021stone} and Chen et al. \cite{chen2022interpretable}, delve into the realm of learning weighted STL (wSTL) formulas~\cite{9309020}.
The aforementioned works utilize smooth approximations that are not sound due to the approximation errors of the robustness calculations. 
Chen et al. \cite{chen2022interpretable} present an alternative approach using the arithmetic-geometric mean (AGM) robustness \cite{mehdipour2019arithmetic}.
The backpropagation of the AGM robustness computation is not scalable.
As a result, it can only handle datasets with a limited number of samples, significantly restricting its applicability. Moreover, \cite{leung2019backpropagation} and \cite{yan2021stone} are not able to learn the structure and the temporal information of the STL formula. \cite{chen2022interpretable} learns the temporal information through a black-box neural network, thereby reducing the interpretability of the learned STL formula.

\vspace{0.1cm}
\textbf{Contributions: }
The contributions of this paper are as follows.
First, we propose TLINet, a novel framework for temporal logic inference using neural networks.
TLINet not only learns the parameters of STL formulas but also captures the structure of the formula and the involved operators and predicates. 
Second, we introduce two innovative approximations for the $\max$ operator in the computation of STL robustness: sparse softmax and averaged max. These approximations are specifically designed to handle temporal and Boolean operators within STL, respectively. Sparse softmax optimizes computational efficiency in temporal contexts, while averaged max provides a succinct representation suitable for Boolean operations. Both approximations are rigorously designed to support gradient-based methods and are accompanied by soundness guarantees. 
Lastly, we apply TLINet to diverse scenarios containing time series data with different properties to show the efficiency and flexibility of our method.

This work extends our previous conference paper~\cite{li2023learning} by making several notable contributions and extensions: (i) A novel vectorized encoding of STL is formulated specifically tailored for training neural networks, (ii) TLINet is able to capture 2nd-order STL specifications, (iii) the learned STL formulas are not confined to disjunctive normal form (DNF)~\cite{hilbert2022principles}, (iv) TLINet can learn the types of operators involved in STL formulas, (v) an additional max approximation is introduced for learning the structure of STL formulas, (vi) additional experiments for STL inference example demonstrating the efficacy of TLINet in learning STL formulas with various structure.

\section{Problem Statement}
In this section, we introduce the syntax and semantics of Signal Temporal Logic (STL), as well as the Temporal Logic Inference problem of inferring STL formulas from time series data.

Let $\mathbf s=[s(0),\ldots,s(l-1)]$ denote a \emph{signal}, where $l$ is the length of the signal, and $s(t)\in\real{d}$ is the state of signal $\mathbf s$ at time $t$.

\subsection{Signal Temporal Logic}
We use STL formulas to specify the temporal and spatial properties of signals. In this paper, we consider a fragment of STL~\cite{maler2004monitoring} without the until operator. 
\begin{definition}\label{def:stl syntax}
The syntax of STL formulas is \cite{maler2004monitoring} defined recursively as:
\begin{equation}
    \phi \Coloneqq \mu \mid \phi_1\land\phi_2 \mid \phi_1\lor\phi_2 \mid \event_{[t_1,t_2]}\phi \mid \always_{[t_1,t_2]}\phi,
\end{equation}
where $\mu$ is a predicate $\mu:= \mathbf{a}^\top\mathbf{s} \sim b$, where $\sim\in\{>,<\}$, $\mathbf{a}\in\real{d}$, $b\in\real{}$. $\phi,\phi_1,\phi_2$ are STL formulas. The Boolean operators $\land,\lor$ are \emph{conjunction} and \emph{disjunction}, respectively. The temporal operators $\event,\always$ represent \emph{eventually} and \emph{always}. $\event_{[t_1,t_2]}\phi$ is true if $\phi$ is satisfied for at least one point $t\in[t_1,t_2]\cap\mathbb Z$, while $\always_{[t_1,t_2]}\phi$ is true if $\phi$ is satisfied for all time points $t\in[t_1,t_2]\cap\mathbb Z$.
\end{definition}

\begin{definition}\label{def:stl semantics}
The quantitative semantics~\cite{donze2010robust}, i.e., the robustness, of an STL formula $\phi$ for signal $\mathbf s$ at time $t$ is defined as:
\begin{subequations}\label{eq:stl semantics}
\begin{align}
    r(\mathbf{s},\mu,t) &= \mathbf{a}^\top \mathbf{s}(t)-b\label{stl-semantics:sub1},\\
    r(\mathbf{s},\land_{i=1}^n \phi_i,t) &= \min_{i=1:n}\{r(\mathbf{s},\phi_i,t)\}\label{stl-semantics:sub3},\\
    r(\mathbf{s},\lor_{i=1}^n \phi_i,t) &= \max_{i=1:n}\{r(\mathbf{s},\phi_i,t)\}\label{stl-semantics:sub4},\\
    r(\mathbf{s},\always_{[t_1,t_2]}\phi,t) &= \min_{t'\in[t+t_1,t+t_2]}r(\mathbf{s},\phi,t')\label{stl-semantics:sub5},\\
    r(\mathbf{s},\event_{[t_1,t_2]}\phi,t) &= \max_{t'\in[t+t_1,t+t_2]}r(\mathbf{s},\phi,t')\label{stl-semantics:sub6}.
\end{align}
\end{subequations}
The robustness is a scalar that measures the degree of satisfaction. The signal $\mathbf s$ is said to satisfy the formula $\phi$, denoted as $\mathbf s \models \phi$, if and only if $r(\mathbf s, \phi, 0) > 0$. Otherwise, $\mathbf s$ is said to violate $\phi$, denoted as $\mathbf s \not\models \phi$.
By convention, we consider zero robustness as violation.
\end{definition}

\subsection{Problem Statement}
\label{sec:tli-problem}
In this paper, we focus on the Temporal Logic Inference (TLI) problem.
The goal is to learn an STL formula from time series data that describes the spatial-temporal properties within the data.
The computed STL formula should classify the data into desired and undesired behaviors.

Let $\cS = \{(\mathbf{s}^i,c^i)\}_{i=1}^N$ be a labeled dataset, where $\mathbf{s}^i$ is the $i^{th}$ signal with label $c^i\in \cC$, and $\cC=\{1,-1\}$ is the set of classes.

\begin{problem}
Given $\cS=\{(s^i,c^i)\}_{i=1}^N$, learn an STL formula $\phi$ that accurately classifies the data into the desired classes, minimizing the misclassification rate (MCR), defined as the ratio of misclassified samples to the total number of samples:
\begin{equation*}
    MCR = \frac{\lvert\{\mathbf{s}^i\mid (\mathbf{s}^i \models \phi \land c^i = -1) \lor (\mathbf{s}^i \not\models \phi \land c^i = 1)\}\rvert}{N}.
\end{equation*}
\end{problem}

\section{Approach Overview}
In this section, we present an overview of our approach, focusing on the integration of Signal Temporal Logic (STL) and neural networks for Temporal Logic Inference (TLI).

According to the grammar in \eqref{def:stl syntax}, STL formulas are composed of operators arranged in a hierarchical manner. Similarly, neural networks consist of layers with interconnected nodes, where higher layers abstract features from lower layers, forming a hierarchical representation. Furthermore, the operators of STL have analogies with the neurons of neural networks. Thus, our approach aims to construct a neural network that can be translated to an STL specification after training.
In the next section, we introduce an encoding for STL formulas that facilitates embedding them into neural networks and enables learning of the formula's structure.

\subsection{Vectorized Signal Temporal Logic}
We define an encoding language for STL templates, called \emph{vectorized Signal Temporal Logic} (vSTL). It is an extension of Parametric Signal Temporal Logic (PSTL) \cite{asarin2012parametric} and weighted Signal Temporal Logic (wSTL) \cite{mehdipour2020specifying}, where binary weights are used to parameterize the structure of the formula, while the spatial parameters $\mathbf{a}$, $b$ are considered continuous parameters. The presented encoding works for discrete-time signals and systems.

\begin{definition}\label{def:vstl syntax}
The syntax of vSTL is derived from the STL syntax using binary weight vectors:
\begin{equation}
    \phi::=\mu \mid \land_i^{\mathbf{w}^b} \phi_i \mid \lor_i^{\mathbf{w}^b} \phi_i \mid \event_{\mathbf{w}^I} \phi \mid \always_{\mathbf{w}^I} \phi,
\end{equation}
where $\mathbf{w}^b=[w^b_i]_{i=1:n}\in\{0,1\}^n$ is a \emph{Boolean vector} associated with a Boolean operator $\land$ and $\lor$, and $w^b_i=1$ if $\phi_i$ is included in the Boolean operation; $\mathbf{w}^I\in\{0,1\}^l$ is a \emph{time vector} associated with temporal operators $\event$ and $\always$, and $I=[t_1,t_2]$ represents the time interval, $w^I_t=1$ if $t_1\leq t \leq t_2$, else $w^I_t=0$. The interpretation is consistent with that of STL by assigning binary weights.
\end{definition}

To enable vectorized computation, we introduce the notion of \emph{robustness vector}.
\begin{definition}[Robustness Vector]
    Given a signal $\mathbf s$, the robustness vector of an STL formula $\phi$ is $\mathbf{r}^v(\mathbf{s},\phi)$ containing the robustness values of $\phi$ at all time steps:
    \begin{equation}
        \mathbf{r}^v(\mathbf{s},\phi) = \bmat{r(\mathbf{s},\phi,0), r(\mathbf{s},\phi,1), \cdots, r(\mathbf{s},\phi,l-1)},
    \end{equation}
    and the robustness vector at time $t$ is $\mathbf{r}^b_{\varphi}(\mathbf{s},t)$ defined as:
    \begin{equation}
        \mathbf{r}^b_{\varphi}(\mathbf{s},t) = \bmat{r(\mathbf{s},\phi_1,t), r(\mathbf{s},\phi_2,t), \cdots, r(\mathbf{s},\phi_n,t)},
    \end{equation}
    where $\varphi$ takes the Boolean operation of children subformulas $\phi_1,\cdots,\phi_n$ based on $\mathbf{w}^b$, i.e., $\varphi=\land_i^{\mathbf{w}^b}\phi_{i}$ or $\varphi=\land_i^{\mathbf{w}^b} \phi_{i}$.
\end{definition}

\begin{definition}\label{def:vstl robustness}
The robustness of a vSTL formula $\phi$ over signal $\mathbf s$ at time $t$ is defined as:
\begin{subequations}
\begin{align}
    r(\mathbf{s},\mu,t) &= \mathbf{a}^\top\mathbf{s}(t)-b\label{vstl:sub5},\\
    r(\mathbf{s},\land_i^{\mathbf{w}^b} \phi_i,t) &= -\max_{\mathbf{w}^b} (-\mathbf{r}^b_{\varphi}(\mathbf{s},t))\label{vstl:sub1},\\
    r(\mathbf{s},\lor_i^{\mathbf{w}^b} \phi_i,t) &= \max_{\mathbf{w}^b} (\mathbf{r}^b_{\varphi}(\mathbf{s},t) ),\label{vstl:sub2}\\
    r(\mathbf{s},\always_{\mathbf{w}^I}\phi,t) &= -\max_{\mathbf{w}^I} (-\mathbf{r}^v(\mathbf{s},\phi) ),\label{vstl:sub3}\\
    r(\mathbf{s},\event_{\mathbf{w}^I}\phi,t) &= \max_{\mathbf{w}^I}(\mathbf{r}^v(\mathbf{s},\phi))\label{vstl:sub4},
\end{align}
\end{subequations}

where
\begin{equation}\label{vstl:max}
\begin{aligned}
    \max_{\mathbf{w}^b} (\mathbf{r}^b_{\varphi}(\mathbf{s},t)) &= \max \left[r(\mathbf{s},\phi_i,t)\right]_{\mathbf{w}^b_i=1},\\
    \max_{\mathbf{w}^I} (\mathbf{r}^v(\mathbf{s},\phi)) &= \max \left[r(\mathbf{s},\phi,t')\right]_{\mathbf{w}^I_{t'}=1},
\end{aligned}
\end{equation}
represent the maxima over the values where the weight vectors are one.
\end{definition}

\begin{proposition}
From a vSTL formula, we can syntactically extract only one equivalent STL formula.
\end{proposition}

\begin{proof}
The time interval of an STL formula can be inferred from $\mathbf{w}^I$ by identifying the indices where $w^I_i=1$, while the subformulas involved in the Boolean operation can be deduced from $\mathbf{w}^b$ by locating the indices where $w^b_i=1$. It is worth noting that for an STL formula, we have an infinite number of vSTL formulas syntactically consistent with it. Since there are STL formulas that differ, but define the same language (set of satisfying signals), i.e., they are semantically equivalent.
\end{proof}

We use vSTL, which is defined on vectors, as it is more suitable for computation and training within neural networks (in particular, the weights $\vw^{I}$ and $\vw^{b}$ can be learned, see Section \ref{sec:modules}). Despite its vector-based representation, the robustness of vSTL formulas remains consistent with traditional STL. One notable advantage of vSTL is its ability to provide detailed information not only about what is included in the operator but also about what is excluded.
For instance, vSTL allows us to infer the subformulas that are not explicitly included in the STL formula. This additional level of information makes vSTL particularly informative and versatile, offering enhanced insights into properties for various applications.

Unless specified otherwise, this paper assumes the signal initiates at time 0, and its robustness is assessed at the same time point.

\subsection{TLINet as an STL formula}
We propose a differentiable Neural Network for the TLI problem, called TLINet.
Each of its layers contains modules for operators defined in Definition~\ref{def:vstl robustness}.
We introduce three types of modules: predicate, temporal, and Boolean modules. 
\begin{itemize}
    \item The predicate module learns the predicate type and spatial parameters.
    \item The temporal module learns the type of temporal operator and temporal parameters.
    \item The Boolean module learns the type of Boolean operator and the structure of the formula.
\end{itemize}

We construct the TLINet by specifying the number of layers, the type of each layer and the number of modules in each layer. After the training process, we decode the parameters of TLINet, translate each module to a composition of STL formulas, and extract the overall STL formula.
An example of a TLINet is shown in Fig.~\ref{fig:example-tlinet}.

\begin{figure}
\centering
\begin{tikzpicture}
\begin{scope}[every node/.style={minimum width=2cm,draw}]
\node[draw=DodgerBlue3,font=\fontsize{8}{8}\selectfont] (layer 1a) {\textsf{predicate module}};
\node[draw=DodgerBlue3,right=3mm of layer 1a,font=\fontsize{8}{8}\selectfont] (layer 1b) {\textsf{predicate module}};
\node[draw=black!30!green,below=5mm of layer 1a,font=\fontsize{8}{8}\selectfont] (layer 2a) {\textsf{Boolean module}};
\node[draw=black!30!green,font=\fontsize{8}{8}\selectfont] (layer 2b) at (layer 2a-|layer 1b) {\textsf{Boolean module}};
\node[draw=orange,below=5mm of layer 2a,font=\fontsize{8}{8}\selectfont] (layer 3a) {\textsf{temporal module}};
\node[draw=orange,font=\fontsize{8}{8}\selectfont] (layer 3b) at (layer 3a-|layer 2b) {\textsf{temporal module}};
\node[draw=orange,below=5mm of layer 3a,font=\fontsize{8}{8}\selectfont] (layer 4a) {\textsf{temporal module}};
\node[draw=orange,font=\fontsize{8}{8}\selectfont] (layer 4b) at (layer 4a-|layer 3b) {\textsf{temporal module}};
\node[draw=black!30!green,below of=layer4a,font=\fontsize{8}{8}\selectfont] (layer 5) at ($(layer 4a)!0.5!(layer 4b)$) {\textsf{Boolean module}};
\end{scope}

\newcommand{\ofst}{5mm}
\begin{scope}[every node/.style={dashed,draw,minimum width=6cm, minimum height=1cm,rounded corners=1mm}]
\node[DodgerBlue3] at ($(layer 1a)!0.5!(layer 1b)$)(frame 1) {};
\end{scope}
\begin{scope}[every note/.style={minimum width=2cm}]
\node[DodgerBlue3,left=0mm of frame 1,font=\fontsize{8}{8}\selectfont] {\textsf{Predicate Layer}};
\end{scope}
\begin{scope}[every node/.style={dashed,draw,minimum width=6cm, minimum height=1cm,rounded corners=1mm}]
\node[black!30!green] at ($(layer 2a)!0.5!(layer 2b)$)(frame 2) {};
\end{scope}
\begin{scope}[every note/.style={minimum width=2cm}]
\node[black!30!green,left=0mm of frame 2,font=\fontsize{8}{8}\selectfont] {\textsf{Boolean Layer}};
\end{scope}
\begin{scope}[every node/.style={dashed,draw,minimum width=6cm, minimum height=1cm,rounded corners=1mm}]
\node[orange] at ($(layer 3a)!0.5!(layer 3b)$)(frame 3) {};
\end{scope}
\begin{scope}[every note/.style={minimum width=2cm}]
\node[orange,left=0mm of frame 3,font=\fontsize{8}{8}\selectfont] {\textsf{Temporal Layer}};
\end{scope}
\begin{scope}[every node/.style={dashed,draw,minimum width=6cm, minimum height=1cm,rounded corners=1mm}]
\node[orange] at ($(layer 4a)!0.5!(layer 4b)$)(frame 4) {};
\end{scope}
\begin{scope}[every note/.style={minimum width=2cm}]
\node[orange,left=0mm of frame 4,font=\fontsize{8}{8}\selectfont] {\textsf{Temporal Layer}};
\end{scope}
\begin{scope}[every node/.style={dashed,draw,minimum width=6cm, minimum height=1cm,rounded corners=1mm}]
\node[black!30!green] at ($(layer 5)$)(frame 5) {};
\end{scope}
\begin{scope}[every note/.style={minimum width=2cm}]
\node[black!30!green,left=0mm of frame 5,font=\fontsize{8}{8}\selectfont] {\textsf{Boolean Layer}};
\end{scope}

\node[circle,fill=orange!70,minimum size=0.8cm] at ($(layer 1a)!0.5!(layer 1b)+(0,1)$) (s) {$\mathbf{s}$};
\node[circle,fill=orange!70,minimum size=0.8cm, below=5mm of layer 5] (r) {$\mathbf{r}$};

\begin{scope}[-latex]
\draw (s) -- (layer 1a);
\draw (s) -- (layer 1b);
\draw (layer 1a) -- (layer 2a);
\draw (layer 1b) -- (layer 2b);
\draw (layer 1a) -- (layer 2b);
\draw (layer 1b) -- (layer 2a);
\draw (layer 2a) -- (layer 3a);
\draw (layer 2b) -- (layer 3b);
\draw (layer 3a) -- (layer 4a);
\draw (layer 3b) -- (layer 4b);
\draw (layer 4a) -- (layer 5);
\draw (layer 4b) -- (layer 5);
\draw (layer 5) -- (r);
\end{scope}
\end{tikzpicture}
\caption{An example of a $5$-layer TLINet.}
\label{fig:example-tlinet}
\end{figure}
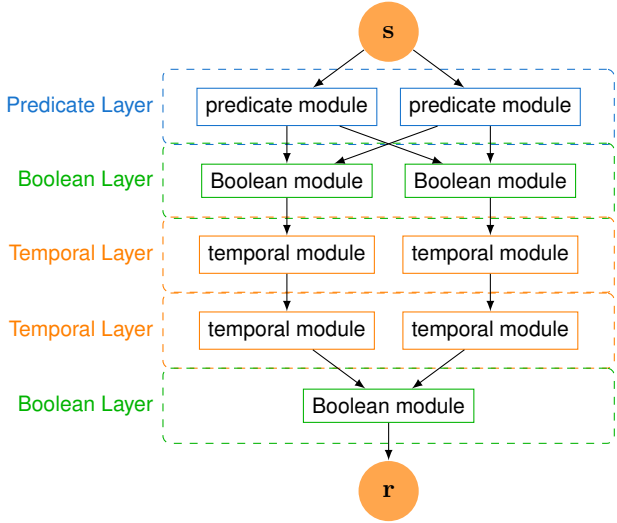

\section{STL Formula Modules}\label{sec:modules}
In this section, we describe the design of modules used in the construction of TLINet. 

\subsection{Predicate Module}\label{sec:predicate}
The predicate module is responsible for computing the robustness of predicates. It functions as a fully connected layer, employing a linear transformation through the weight $\mathbf{a}\in\mathbb{R}^d$ and bias $b \in \mathbb{R}$ on the input. Given the input signal $\mathbf{s} = \bmat{s(0),s(1),\cdots,s(l-1)}\in \mathbb{R}^{d \times l}$, the output is a robustness vector of predicate $\mu$ denoted as $\mathbf{r}^v(s,\mu) = \bmat{r(s,\mu,0),r(s,\mu,1),\cdots,r(s,\mu,l-1)} \in \mathbb{R}^{l}$, where $r(s,\mu,t)=\mathbf{a}^\top s(t)-b$.
The predicate can take the form of axis-aligned by setting some elements of $\mathbf{a}$ to $0$. Figure~\ref{fig:predicate module} shows the computation graph of the predicate module. Example~\ref{predicate example} illustrates the interpretation of the predicate's parameters and provides a visualization of the output generated by the predicate module.
\begin{figure}[ht]
    \centering
    \includegraphics[scale=0.6]{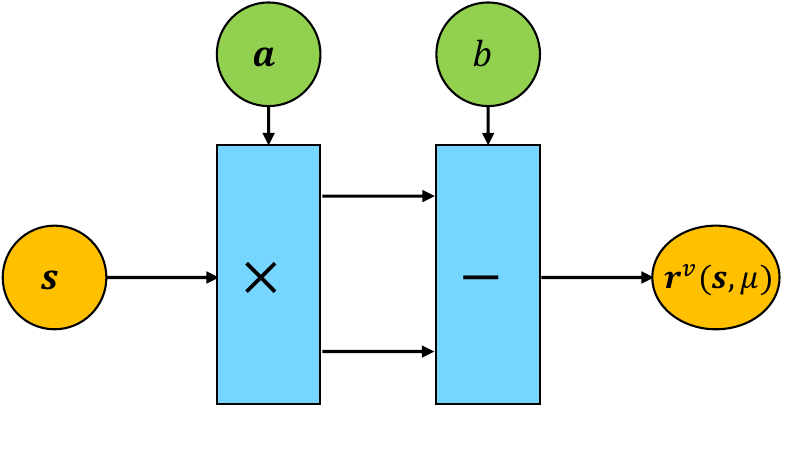}
    \caption{The computation graph of the predicate module, where $\mathbf{a}$ and $b$ are parameters of the module.}
    \label{fig:predicate module}
\end{figure}

\begin{example}\label{predicate example}
    Consider a signal $\mathbf{s}$ and a predicate module with weight $a=-1$ and bias $b=-0.1$, the corresponding predicate is $\mu:=s(t)< 0.1$, Figure \ref{fig:predicate example} shows the input (green) and the output (blue) in time sequences.
\end{example}
\begin{figure}[ht]
    \centering
    \includegraphics[scale=0.5]{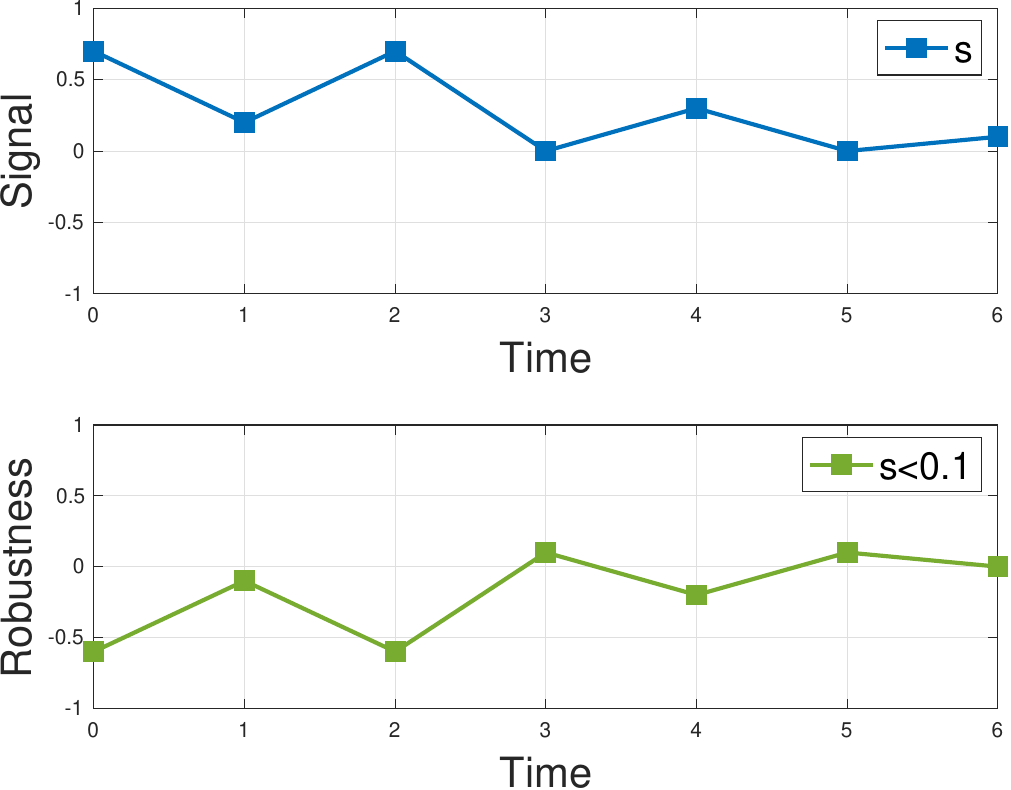}
    \caption{The sequence of robustness $r(t)=-s(t)+0.1$ for predicate $\mu:=s(t)< 0.1$ given a signal $\mathbf s$. The positive and negative robustness implies the degree of satisfaction and violation of the signal to the predicate, respectively.}
    \label{fig:predicate example}
\end{figure}

\subsection{Encoding operator type}\label{sec:variable}
Learning template-free STL formulas involves learning the temporal and Boolean operators defined in \eqref{vstl:sub1}-\eqref{vstl:sub4}. We introduce a binary variable $\kappa$ to determine the operator and generalize the max operation as follows:
\begin{equation}\label{eq:variable}
    r(s,\phi,t) = \kappa\max_{\mathbf{w}} \kappa\mathbf{r},
\end{equation}
where $\kappa$ acts as a switch controlling which operator is applied. For instance, for a temporal operator, if $\kappa=1$, it represents the eventually ($\event$) operator, if $\kappa=-1$, it represents the always ($\always$) operator. Similarly, for a Boolean operator, if $\kappa=1$, it represents the disjunction ($\lor$) operator, if $\kappa=-1$, it represents the conjunction ($\land$) operator.

Instead of directly learning the binary variable $\kappa$, we opt for a continuous parameterization approach. We introduce a real-valued variable $p_{\kappa}$, which represents the likelihood of $\kappa$ being $1$. Inspired by \cite{srinivas2017training}, we consider $\kappa$ as sampled from a Bernoulli distribution based on $p_{\kappa}$, ensuring determinism through the \emph{maximum-likelihood draw}.

Given a set $\cX=\{X_0,X_1\}$ with exactly two elements, define $\bernoulli_{\cX}(p)$ as a Bernoulli distribution such that if $x\sim \bernoulli_{\cX}(p)$ then
\begin{equation}
\begin{aligned}
P(x=X_0)&=p,\\
P(x=X_1)&=1-p.
\end{aligned}
\end{equation}
Then we define the corresponding \emph{maximum likelihood draw} distribution $\bernoulli_{ML,\cX}(p)$ such that if $x\sim \bernoulli_{ML,\cX}(p)$, then
\begin{equation}
    \begin{aligned}
        P(x=X_0)&=1 &\textrm{ if } 0.5\leq p \leq 1,\\
        P(x=X_1)&=1 &\textrm{ if } 0\leq p<0.5. 
    \end{aligned}
\end{equation}

Given the set $\cX_{\kappa}=\{1,-1\}$ with $X_0=1$, $X_1=-1$, and the probability $p_{\kappa}$, we have
\begin{equation}
    \kappa \sim \bernoulli_{ML,\cX_{\kappa}}(p_{\kappa}).
\end{equation}
The gradient of this sampling step is computed using the straight-through estimator\cite{bengio2013estimating}. To maintain the validity of $p_{\kappa}$, we use a clipping function to confine it within the range $\left[0,1\right]$\cite{srinivas2017training}:
\begin{equation}\label{eq:clip}
clip(x)=
    \begin{cases}
      1 & \textrm{if } x\geq 1,\\
      0 & \textrm{if } x\leq 0,\\
      x & \textrm{otherwise}.
    \end{cases}
\end{equation}

We can obtain $\kappa$ from $p_{\kappa}$ using:
\begin{equation}
    \kappa \sim \bernoulli_{ML,\cX_{\kappa}}(clip(p_{\kappa})).
\end{equation}

\subsection{Temporal Module}\label{sec:temporal}
We define the time vector $\mathbf{w}^I$ using time variables $t_1$, $t_2$ from the interval $I=[t_1,t_2]$. To facilitate this, we introduce the concept of \emph{time function}, denoted as $T(t_1,t_2)$, which generates the output $\mathbf{w}^I$.
\begin{definition}[Time Function]\label{time function}
Given time instants $t_1,t_2$ within the range $0\leq t_1 \leq t_2\leq l-1$, the time function $T(t_1,t_2):\mathbb{R}\times\mathbb{R}\to \{0,1\}^l$ is defined as:
    \begin{equation}
    T(t_1,t_2) = \mathbf{w}^I,
    \end{equation}
    where each element of $\mathbf{w}^I$ is given by:
    \begin{equation}
    \begin{aligned}
        w^I_t &= \begin{cases}
            1, & t_1\leq t \leq t_2\\
            0, & 0\leq t <t_1 \lor t_2 <t\leq l-1,
        \end{cases}  
    \end{aligned}
    \label{eq:time_function}
    \end{equation}
\end{definition}
\begin{remark}
    The time horizon $l$ is determined by the length of the signals.
\end{remark}
\begin{definition}[ReLU Activation Function]
The ReLU activation function is defined as:
\begin{equation}
\begin{aligned}
&\relu(x)=\begin{cases}
  x, & \textrm{if } x>0,\\
  0, & \textrm{otherwise}.
\end{cases}
\end{aligned}
\end{equation}
where $x \in \mathbb{R}$.
\end{definition}
In this paper, we adopt a specific time function utilizing the ReLU activation function, defined as:
\begin{equation}
\begin{aligned}
    &T(t_1,t_2)\\
    =&\frac{1}{\eta}\min \Bigl(\relu(\mathbf{n}-\mathbf{1}(t_1-\eta))-\relu(\mathbf{n}-\mathbf{1}t_1),\\
    &\relu(-\mathbf{n}+\mathbf{1}(t_2+\eta))-\relu(-\mathbf{n}+\mathbf{1}t_2)\Bigr),
\end{aligned}
\label{eq:relu time function}
\end{equation}
where $\mathbf{n}=\left[0,1,...,l-1\right]\in\naturals{l}$ is a vector containing $l$ consecutive integers starting from $0$ to $l-1$; $\mathbf{1}\in\{1\}^l$ is a vector with all elements equal to $1$, and $\eta \in\real{}_{> 0}$ is a hyperparameter controlling the slope steepness of $T(t_1, t_2)$.

The ReLU activation function in \eqref{eq:relu time function} can be replaced by other similar activation functions such as the sigmoid activation function or tanh activation function.

Figure \ref{fig:time function} illustrates an example of the time function \eqref{eq:relu time function} with $\eta=\{1,0.5,0.1\}$ for the time interval $I=[4,8]$ and signal length $l=13$. The resulting output $\mathbf{w}^I$ is:
\begin{equation}
\mathbf{w}^I = \bmat{0,0,0,0,1,1,1,1,1,0,0,0,0}.
\end{equation}

The time function requires only two parameters $t_1$, $t_2$ to generate $\mathbf{w}^I$, no matter how long the signal is.

\begin{figure}
    \centering
    \includegraphics[scale=0.5]{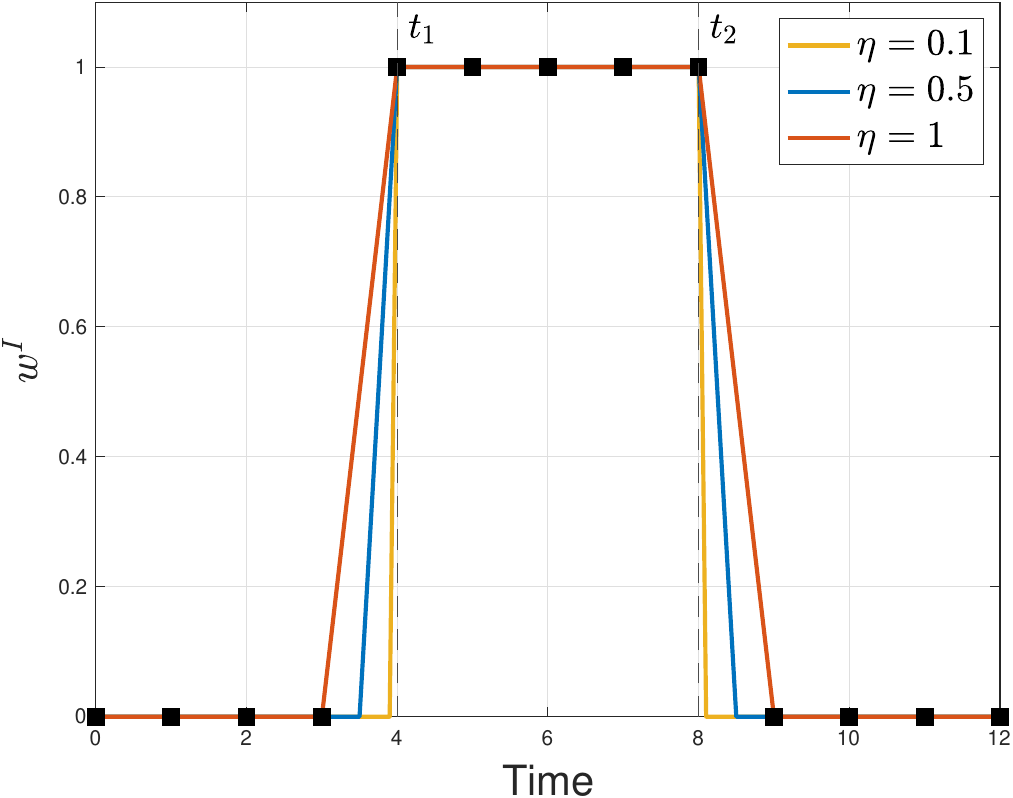}
    \caption{An example of time function. The time interval is $I=[4,8]$. The length of the signal is $13$. The time function with $\eta=0.1,0.5,1$ is shown in yellow, blue, and red, respectively.}
    \label{fig:time function}
\end{figure}

Since STL formulas are defined recursively, a robustness vector is needed for subsequent operations. Therefore, both the input and output of the temporal module must be robustness vectors. However, for a formula $\varphi = \event_{[t_1, t_2]} \phi$, computing the robustness vector $\mathbf{r}^v(\mathbf{s},\varphi)$ at $t \in [0, l-1]$ requires values of $r(\mathbf{s}, \phi, t')$ with $t' \in [t+t_1, t+ t_2]$. Since $0\leq t_1 \leq t_2 \leq l-1$, it follows that the computation needs an input vector of robustness values of length $2l-1$. To resolve this, we introduce a technique called \emph{robustness padding} to lengthen the input robustness vector, thereby facilitating the computation for a valid output robustness vector.

\begin{definition}[Robustness padding]
Given a robustness vector $\mathbf{r}^v(\mathbf{s},\phi)= \bmat{r(\mathbf{s},\phi,0), r(\mathbf{s},\phi,1), \cdots, r(\mathbf{s},\phi,l-1)}$, the robustness padding vector $\mathbf{r}^v_p(\mathbf{s},\phi)$ is defined as:
\begin{equation}
    \mathbf{r}^v_p(\mathbf{s},\phi) = [\underbrace{\rho,\cdots,\rho}_{l-1}],
\end{equation}
where
\begin{equation}
    \rho = \min_{t\in\left[0,l-1\right]} r(\mathbf{s},\phi,t).
    \label{eq:padding}
\end{equation}
The padded robustness vector $\mathbf{p}^v(\mathbf{s},\phi)$ is
\begin{equation}
\begin{aligned}
    \mathbf{p}^v(\mathbf{s},\phi) &= \bmat{\mathbf{r}^v(\mathbf{s},\phi),\mathbf{r}^v_p(\mathbf{s},\phi)}.
\end{aligned}
\end{equation}
\end{definition}

Given that the padding value $\rho$ represents the minimum of the robustness values, it is subsequently ignored through the $\max$ operation. The robustness padding in~\eqref{eq:padding} is applied prior to the $\max$ function, ensuring that subsequent robustness computations remain unaffected by the padding. See Example~\ref{padding example} for further clarification.

\begin{example}\label{padding example}
    Consider a 2nd-order STL specification $\phi_2 = \always_{[0,3]}\phi_1$ with $\phi_1=\event_{[1,4]}(s>0.1)$.
    Given a signal $\mathbf{s}$ of length~$8$, we first compute $\mathbf{r}^v(\mathbf{s},\mu)\in\real{1\times 8}$ for predicate $\mu:=s>0.1$, use the robustness padding vector $\mathbf{r}^v_p(\mathbf{s},\mu)\in\real{1\times 7}$ to extend the predicate vector, then we can compute the robustness vector $\mathbf{r}^v(\mathbf{s},\phi_1)$.
    Figure \ref{fig:padding} visualizes the procedure from the signal $\mathbf s$ to the robustness vector $\mathbf{r}^v(\mathbf{s},\phi_1)$. The robustness value of $\phi_2$ is
    \begin{equation}
        r(\mathbf{s},\phi_2,0)=-\max_{\mathbf{w}^I} (-\mathbf{r}^v(\mathbf{s},\phi_1))=0.5,
    \end{equation}
    where $\mathbf{w}^I = \bmat{1,1,1,1,0,0,0,0}$.
\end{example}

\begin{figure}[ht]
    \centering
    \includegraphics[scale=0.5]{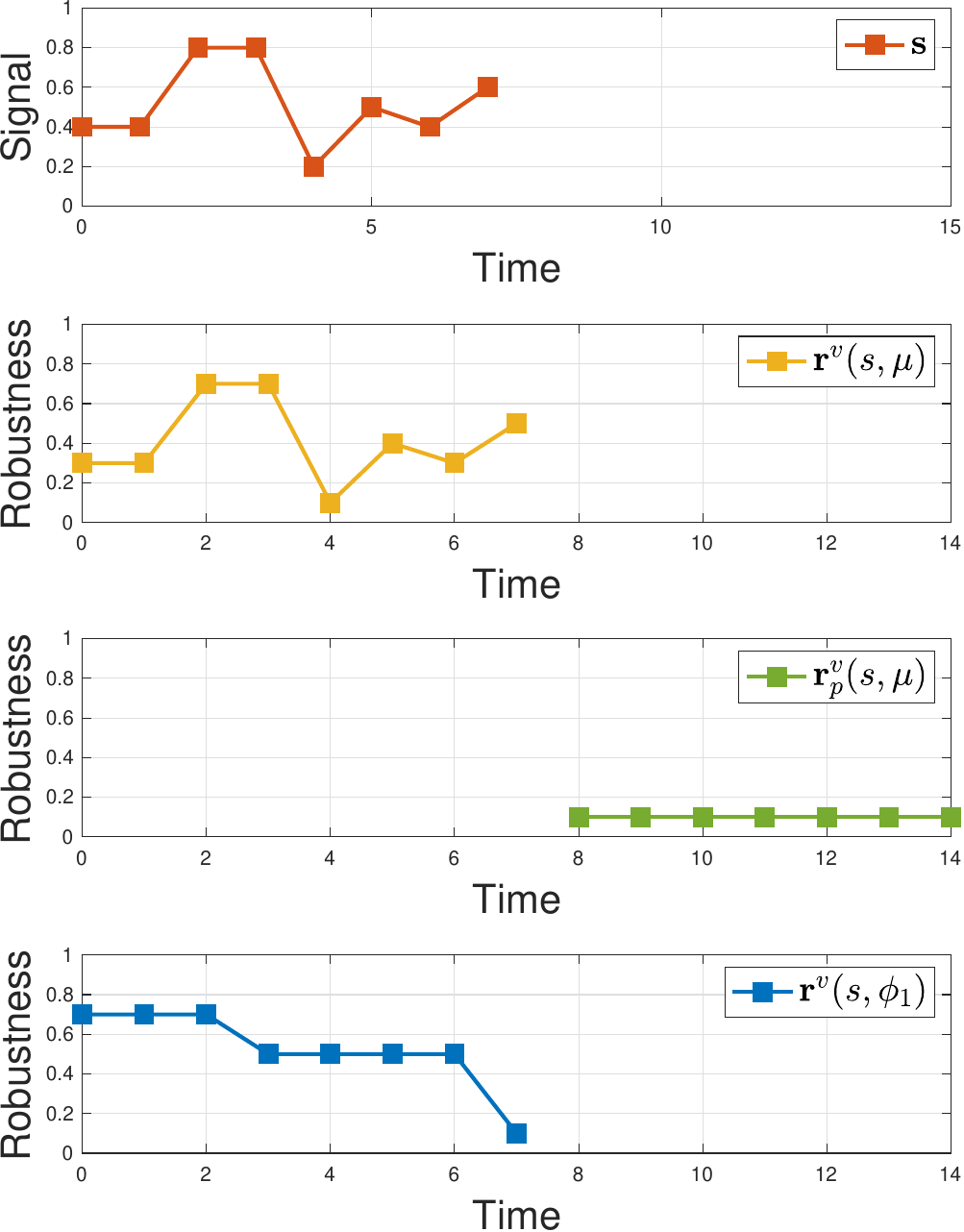}
    \caption{The procedure from signal $\mathbf s$ to the robustness vector $\mathbf{r}^v(\mathbf{s},\phi_1)$. The signal $\mathbf s$ is shown in red. The robustness vector of predicate $\mathbf{r}^v(\mathbf{s},\mu)$ is shown in yellow. The robustness padding vector of predicate $\mathbf{r}^v_p(\mathbf{s},\mu)$ is shown in green. The robustness vector $\mathbf{r}^v(\mathbf{s},\phi_1)$ is shown in blue. }
    \label{fig:padding}
\end{figure}

By integrating both the time function and robustness padding technique within the temporal module, we ensure consistent and accurate computation of robustness vectors. The overview of the temporal module structure is shown in Figure \ref{fig:temporal module}.

\begin{figure}[ht]
    \centering
    \includegraphics[scale=0.6]{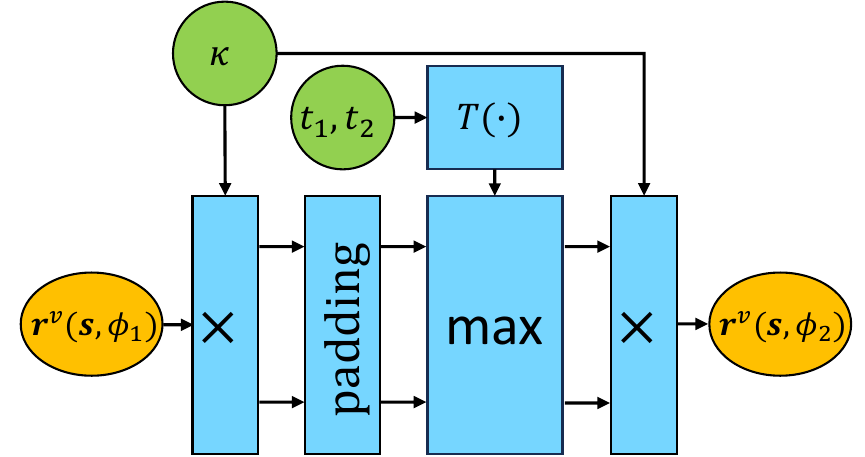}
    \caption{The structure of the temporal module.}
    \label{fig:temporal module}
\end{figure}

\subsection{Boolean Module}\label{sec:Boolean}
The Boolean module processes robustness vectors of multiple subformulas $\phi_1,\phi_2,\cdots,\phi_n$ arranged in a matrix format:
\begin{equation}
\bmat{\mathbf{r}^v(\mathbf{s},\phi_1),&\mathbf{r}^v(\mathbf{s},\phi_2),\cdots,\mathbf{r}^v(\mathbf{s},\phi_n)}.
\end{equation}
The output of the Boolean module is a robustness vector $\mathbf{r}^v(\mathbf{s},\varphi) = \bmat{r(\mathbf{s},\varphi,0), r(\mathbf{s},\varphi,1), \cdots, r(\mathbf{s},\varphi,l-1)}$, where $\varphi$ takes the Boolean operation of $\phi_1,\cdots,\phi_n$ based on $\mathbf{w}^b$.
This binary vector $\mathbf{w}^b$ determines the inclusion or exclusion of each subformula in the Boolean operation, allowing for flexible combinations. The structure of the Boolean module follows the computation described in \eqref{eq:variable}. Figure \ref{fig:boolean module} visually illustrates the structure and operation of the Boolean module within the framework.

\begin{figure}[ht]
    \centering
    \includegraphics[scale=0.6]{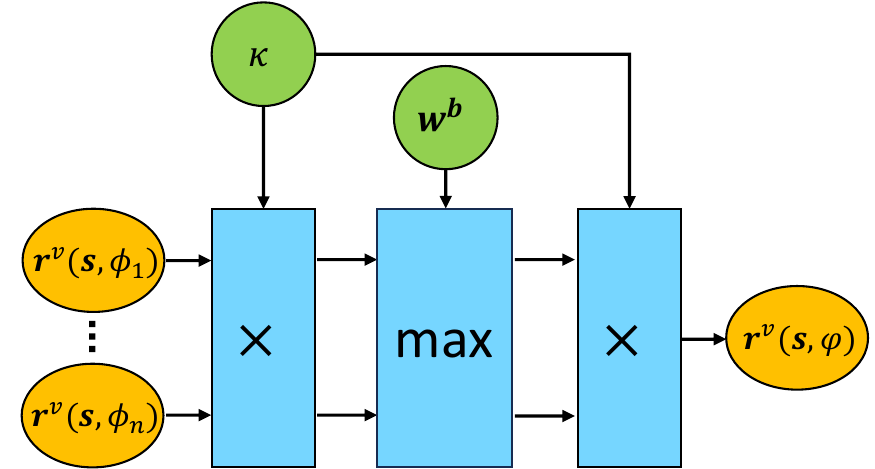}
    \caption{The structure of the Boolean module.}
    \label{fig:boolean module}
\end{figure}

Similarly to the learning approach for $\kappa$, we extend the concept to the binary vector $\mathbf{w}^b=\bmat{w^b_1,w^b_2,\cdots,w^b_n}\in\{0,1\}^n$. We consider $w^b_i$ as a binary variable sampled from a Bernoulli distribution with probability $p^b_i$ through maximum likelihood draw:
\begin{equation}
    w^b_i \sim \bernoulli_{ML,\cX_{\mathbf{w}^b}}(clip(p^b_i)),
\end{equation}
where $\cX_{\mathbf{w}^b}=\{1,0\}$ with $X_0=1$ and $X_1=0$.

\section{Max Approximation Methods}
In this section, we introduce methods for approximating the max operation in \eqref{eq:variable}. The $\max$ operation in \eqref{eq:variable} often result in numerous zero gradients during backpropagation, causing gradient vanishing problem, which significantly slows down or halts the training progress of neural networks. To address this issue, we propose two types of approximation methods that are highly adaptable to the training of temporal logic-based neural networks.

\subsection{Desired Properties}
The max approximation methods for TLINet need to possess certain properties. We evaluate these properties from both a learning-based perspective and their suitability for STL~\cite{varnai2020robustness}. 
First, it is crucial for these methods to enable the utilization of gradient-based techniques.

\begin{property}[Differentiable Almost Everywhere]\label{property:diff}
A function is differentiable almost everywhere if it is differentiable everywhere except on a set of measure zero~\cite{Halmos1974-HALMTQ}.
\end{property}

\begin{property}[Gradient Stability]\label{property:gradient}
A function exhibits gradient stability if it does not suffer from gradient vanishing or exploding problems.
\end{property}

The approximation methods must satisfy Property \ref{property:diff} and \ref{property:gradient} to maintain a meaningful gradient flow, allowing for stable and effective optimization for neural networks.

Next, these methods yield robustness values based on \eqref{eq:variable}. The robustness value's sign must explicitly convey whether it satisfies the corresponding STL specification. Hence, the assurance of soundness is crucial for TLINet.

\begin{property}[Soundness]\label{property:soundness}
Let $M(\mathbf{x},\mathbf{w})$ denote a function for computing the maximum of $\mathbf{x}$ given $\mathbf{w}$. We say $M$ is sound if
\begin{equation}
    \begin{aligned}
        \max_{\mathbf{w}}(\mathbf{x})> 0 \iff M(\mathbf{x},\mathbf{w})> 0,\\
        \max_{\mathbf{w}}(\mathbf{x})\leq 0 \iff M(\mathbf{x},\mathbf{w})\leq 0.
    \end{aligned}
\end{equation}
\end{property}

\subsection{Softmax}
A general approximation method for STL max operation $\displaystyle\max_{\mathbf{w}}(\mathbf{x})$ is the softmax function~\cite{yan2021stone} defined as:
\begin{equation}\label{softmax}
\begin{aligned}
    s(\mathbf{x},\mathbf{w})=\frac{\sum_{i=1}^{n} x_i w_i e^{\beta x_i}}{\sum_{i=1}^{n} w_i e^{\beta x_i}}= \sum_{i=1}^{n} x_i q^s_i,
\end{aligned}
\end{equation}
where $\mathbf{x}\in\real{n}$ is the input vector, $\mathbf{w}\in\{0,1\}^n$ is the time vector or Boolean vector, $\beta\in\real{}_{> 0}$ is a scaling parameter.
\begin{proposition}\label{prop:limit} Let $\cI\subset \{1,\ldots,N\}$ be a subset of indices, and $\bar{\cI}$ be its complement in $\{1,\ldots,N\}$. If we keep the values of $x_\cI$ fixed and let the values of $x_{\bar{\cI}}$ go to $-\infty$, the corresponding weights $q^s_{\bar{\cI}}$ will go to zero, i.e., 
\begin{equation}
    \lim_{x_{\bar{\cI}}\to-\infty} q^s_{\bar{\cI}}=0.
\end{equation}
\end{proposition}
The softmax function takes the weighted sum of input values $x_i$'s with weights $q_i^s$'s. In this context, smaller input values correspond to smaller weights, reducing their contribution to the computed maximum. Such calculations have been employed as activation functions in temporal logic-based neural networks in~\cite{leung2021back,yan2021stone}.
Nevertheless, the softmax function does not satisfy Property \ref{property:soundness}, i.e., the assurance of soundness is not guaranteed~\cite{maler2004monitoring}, introducing the possibility of inaccuracies in results and misinterpretation of STL specifications.

In the following sections, we introduce two approximation techniques for the max function that satisfy all the desired properties. Thus, these methods can be seamlessly employed in neural networks based on temporal logic.

\subsection{Sparse Softmax}
To improve the softmax function, we propose the \emph{sparse softmax function} that can guarantee the soundness property. Intuitively, \Cref{prop:limit} shows that when the values of $x_{\bar{\cI}}$ are sufficiently smaller than those in $x_{\cI}$, they will have a negligible influence on the result. Hence, we refer to it as the ``sparse" softmax function. The sparse softmax function $S(\mathbf{x},\mathbf{w})$ is defined through the following sequence of operations:
\begin{subequations}\label{sparsemax}
\begin{align}
    x_i' &= x_i w_i\label{sparsemax:sub1}\\
    x_{m} &= \begin{cases}
    \lvert \displaystyle\max_{i} (x_i')\rvert & \textrm{if } \lvert \displaystyle\max_i(x_i')\rvert \neq 0,\\
    1 & \textrm{otherwise.}
    \end{cases}\label{sparsemax:sub2}\\
    x_i'' &= \frac{h x_i'}{x_{m}}\label{sparsemax:sub3},\\
    q_i &= \frac{e^{\beta x_i''}}{\sum_i e^{\beta x_i''}}\label{sparsemax:sub4},\\
    S(\mathbf{x},\mathbf{w}) &= \frac{\sum_{i=1}^{n} x_i w_i q_i}{\sum_{i=1}^{n} w_i q_i} = \sum_{i=1}^{n}x_i q^S_i \label{sparsemax:sub5},
\end{align}
\end{subequations}
where $h\in\real{}_{>0}$ and $\beta\in\real{}_{>0}$ are hyperparameters. We scale $x_i$ to $x_i''$ through \eqref{sparsemax:sub1} to \eqref{sparsemax:sub3}, then transfer it into a probability distribution $q^S$ such that $\sum_i q_i^S = 1$. The sparse softmax function computes the weighted sum of $x_i$'s using weights $q_i^S$'s. 

\begin{proposition}\label{prop:sound}
The sparse softmax function is sound if the hyperparameters $\beta$, $h$ satisfy $h e^{\beta h}>\frac{(n-1) e^{-1}}{\beta}$.

The proof is shown in Appendix \ref{appendix:soundness}.
\end{proposition}
If the condition in Proposition \ref{prop:sound} is satisfied, the weights of some elements in $\mathbf{x}$ are small enough to be ignored compared to others. See Example \ref{example:sparse softmax} for a concrete example.

\begin{example}\label{example:sparse softmax}
Let a signal $\mathbf{s}=\bmat{2,1.1,0.9,0,-1}$. Consider the STL specification $\phi=\event_{[1,4]} (s>1)$. From the time interval $I=[1,4]$, the time vector $\mathbf{w}^I=\bmat{0,1,1,1,1}$. The robustness vector of predicate $\mu:=s>1$ is $\mathbf{r}^v(\mathbf{s},\mu)=\bmat{1,0.1,-0.1,-1,-2}$. The true robustness computed from \eqref{vstl:sub4} is $r(s,\phi,0)=0.1>0$. Choosing $\beta=1$, the robustness computed from softmax function is 
\begin{equation}
    s(\mathbf{r},\mathbf{w}^I)=\sum_{i=0}^{4} r_i q^s_i = -0.246<0
\end{equation}
where $\mathbf{q}^s = \bmat{0,0.440,0.360,0.146,0.054}$.

The robustness computed using our sparse softmax function is  
\begin{equation}
\begin{aligned}
\mathbf{r}' &= \bmat{0, 0.1, -0.1, -1, -2},\\
\mathbf{r}'' &= \bmat{0, 1, -1, -10, -20},\\
\mathbf{q}^S &= \bmat{0, 0.88, 0.12, 0, 0}\\
S(\mathbf{r},\mathbf{w}^I) &= \sum_{i=0}^{4} r_i q^S_i = 0.076>0,\\
\end{aligned}
\end{equation}
with $h=1$ to satisfy $he^{\beta h}>\frac{4e^{-1}}{\beta}$.
\end{example}

Compared to $\mathbf{q}^s$, some elements of $\mathbf{q}^S$ are zero, redistributing more weights onto other elements. Thus, the sparse softmax function can provide valid robustness, whereas the softmax function may fail to do so. In a classification problem, an algorithm using the softmax function will misclassify a signal $\mathbf{s}$ as violating $\phi$, even when $\mathbf{s}$ satisfies $\phi$.

\subsection{Averaged Max}
The max operation $y=\displaystyle\max_{\mathbf{w}}(\mathbf{x})$ in vSTL is a function where $\mathbf{x}\in\real{n}$, $\mathbf{w}\in\{0,1\}^n$, $w_i$ is an independent gating variable to determine whether $x_i$ should be included in the function $\max(\cdot)$ or not. Inspired by \cite{srinivas2017training}, $w_i$ can be interpreted as random variables governed by a Bernoulli distribution with probability $p_i$ such that
\begin{equation}
    \begin{aligned}
        P(w_i=1)&=p_i,\\
        P(w_i=0)&=1-p_i,
    \end{aligned}
\end{equation}
where $p_i$ indicates the probability of including $x_i$ in the function $\max(\cdot)$. In this case, $y$ becomes a random variable. We can compute the expectation of $y$ as:
\begin{equation}\label{eq:naive-E}
    \begin{aligned}
    \E(y) &= \E\big(\max_{\mathbf{w}}(\mathbf{x})\big)\\
    &= \E\big(\max(\bmat{x_1,...,x_n}\mid \bmat{w_1,...,w_n})\big)\\
    &= \sum_{\mathbf{w}\in\{0,1\}^n} \big(\max(\{x_i\}_{\{i\in[1,n]:w_i=1\}})\prod_{j=1}^{n} P(w_j)\big).
    \end{aligned}
\end{equation}
Theoretically, $\E(y)$ involves $2^n$ terms, making its computation resource-intensive. To address this, we propose a sorting trick to decrease the computational complexity of $\E(y)$.

We first sort $\bmat{x_1,\cdots,x_n}$ into $\bmat{x_1',\cdots,x_n'}$ such that $x_1'\geq x_2'\geq \cdots\geq x_n'$. Let $m:\cI\to \cI_s$ be the permutation such that $x_i=x'_{m(i)}$ for $i\in\cI$ and $m(i)\in\cI_s$. Note that $\max(x_{\cI})=\max(x'_{\cI_s})$ for any $\cI_s$. The \emph{averaged max function}, i.e., the expectation of $y$ can be written as:
\begin{equation}\label{eq:expand-E}
    \begin{aligned}
    \E(y) &= \E\big(\max(\bmat{x_1',\cdots,x_n'}|[w_1',...,w_n'])\big)\\
    &= x'_1 p'_1(1-p'_2)\cdots(1-p'_n) +\cdots +x'_1 p'_1p'_2\cdots p'_n\\
    &+ x'_2 (1-p'_1)p'_2\cdots(1-p'_n) +\cdots +x'_2(1-p'_1) p'_2\cdots p'_n\\
    &+\cdots+x'_n (1-p'_1)(1-p'_2)\cdots p'_n\\
    &= x'_1 p'_1+x'_2 p'_2(1-p'_1) +\cdots+x'_n p'_n \prod_{j=1}^{n-1}(1-p'_j)\\&=\sum_{i=1}^n x'_i p'_i \prod_{j=1}^{i-1} (1-p'_j),
    \end{aligned}
\end{equation}
where $p_i=p'_{m(i)}$, $w_i=w'_{m(i)}$. The computational complexity of $\E(y)$ becomes $O(n\log n)$. See Example \ref{example:averagedmax} for a concrete example.

\begin{example}\label{example:averagedmax}
Let the robustness vector be $\mathbf{r}=\bmat{r_0,r_1,r_2}$ and $r_2>r_1>r_0$, $\mathbf{w} = \bmat{w_0,w_1,w_2}$. The expectation of $y=\displaystyle\max_{\mathbf{w}}(\mathbf{x})$ is
{\allowdisplaybreaks
\begin{align*} 
&\E\big(\max(\bmat{r_0,r_1,r_2}|\bmat{w_0,w_1,w_2})\big)\\
=& \max(r_0)P(\mathbf{w}=[1,0,0]) + \max(r_1)P(\mathbf{w}=[0,1,0])\\
&+\max(r_0,r_1)P(\mathbf{w}=[1,1,0])+ \max(r_2)P(\mathbf{w}=[0,0,1])\\
&+ \max(r_0,r_2)P(\mathbf{w}=[1,0,1])\\
&+ \max(r_1,r_2)P(\mathbf{w}=[0,1,1])\\
&+\max(r_0,r_1,r_2)P(\mathbf{w}=[1,1,1])\\
=& r_0p_0(1-p_1)(1-p_2) + r_1(1-p_0)p_1(1-p_2)\\
&+ r_1p_0p_1(1-p_2)+ r_2(1-p_0)(1-p_1)p_2\\
&+ r_2p_0(1-p_1)p_2 + r_2(1-p_0)p_1p_2+ r_2p_0p_1p_2\\
=& r_2p_2 + r_1p_1(1-p_2) + r_0p_0(1-p_1)(1-p_2).
\end{align*}}
\end{example}

Note that when all ${p_i}$'s converge to either $0$ or $1$, the expected value $\E(y)$ equals to $y$, ensuring the soundness property. To accommodate this, we introduce a bi-modal regularizer from \cite{srinivas2017training}, which encourages values to approach either $0$ or $1$. The bi-modal regularizer for the averaged max function is
\begin{equation}\label{eq:avm reg}
    l_{avm} = \sum_{i=1}^{n} p_i(1-p_i).
\end{equation}

\subsubsection{Averaged Minmax}
If we consider $\kappa$ and $\vw$ as Bernoulli random variables, the output of \eqref{eq:variable} will also be a random variable. 
We therefore propose to use its expected value, defined as 
\begin{equation}\label{eq:variable averaged max}
        z = \mathbb{E}[\kappa\max_{\vw} (\kappa\mathbf{x})],
\end{equation}
where $\mathbf{x}\in\real{n}$, and the expectation is taken over $\mathbf{w}\in\{0,1\}^{n}$ and $\kappa\in\{1,-1\}$ distributed as Bernoulli variables. 
We have defined $\kappa$ in Section \ref{sec:variable} that:
\begin{equation}
\begin{aligned}
    P(\kappa=1) &= p_\kappa,\\
    P(\kappa=-1) &= 1-p_\kappa.
\end{aligned}
\end{equation}

The averaged minmax function, i.e., the expectation of $z$, becomes
\begin{equation}
    \begin{aligned}
    \E(z) &= p_{\kappa}(x'_1 p'_1+ ...+x'_{n} p'_{n} \prod_{j=1}^{n-1}(1-p'_j))\\
    &+(1-p_{\kappa})(x'_{n} p'_{n}+ ...+x'_1 p'_1 \prod_{j=2}^{n}(1-p'_j)),
    \end{aligned}
\end{equation}
where the first term is the expectation of taking the max operation while the second term is the expectation of taking the min operation.

Similar to the averaged max function, ${p_i}$'s and $p_k$ both need to converge to $0$ or $1$ to guarantee soundness property. Therefore, the regularizer for the averaged minmax function is
\begin{equation}\label{eq:kavm reg}
    l_{kavm} = p_k(1-p_k) + \sum_{i=1}^{n} p_i(1-p_i).
\end{equation}
The averaged minmax function excludes the use of the straight-through estimator, thus making the backpropagation smoother.

\subsection{Comparative Analysis of Sparse Softmax and Averaged Max}
We have introduced two newly developed approximation methods for the STL max operation: sparse softmax and averaged max. Both methods satisfy the desired properties, yet they are suited for different operators within STL.

The sparse softmax function is particularly suitable for temporal operators. This is because the output of the time function is inherently a binary vector, aligning well with the sparse softmax function. Conversely, employing the averaged max function for temporal operators requires learning the probability of weights across all time points. These weights must converge to binary values of $0$ or $1$, presenting challenges in training. Directly applying the output of the time function to the averaged max function can lead to issues with gradient stability, akin to employing a hard max function.

On the other hand, for Boolean operators, the averaged max function is preferred since it can naturally accept the probability of weights of subformulas, obviating the need for the straight-through estimator. Moreover, the averaged max learns to ``select" elements taken for the max operation while simultaneously approximating the maximum.

By understanding the distinct advantages and limitations of each method, practitioners can make informed decisions regarding their applications within STL-based neural networks.


\begin{figure*}
\begin{subfloat}[Initialize TLINet structure.]{
\begin{tikzpicture}
\begin{scope}[every node/.style={minimum width=2cm,draw}]
\node[draw=DodgerBlue3] (layer 1a) {$(a,b)$};
\node[draw=DodgerBlue3,right= 3mm of layer 1a] (layer 1b) {$(a,b)$};
\node[draw=orange,below=5mm of layer 1a] (layer 2a) {$(\kappa,t_1,t_2)$};
\node[draw=orange] (layer 2b) at (layer 2a-|layer 1b) {$(\kappa,t_1,t_2)$};
\node[draw=orange,below=5mm of layer 2a] (layer 3a) {$(\kappa,t_1,t_2)$};
\node[draw=orange] (layer 3b) at (layer 3a-|layer 2b) {$(\kappa,t_1,t_2)$};
\node[draw=black!30!green,below=5mm of layer 3a] (layer 4a) {$(\kappa,w)$};
\node[draw=black!30!green] (layer 4b) at (layer 4a-|layer 3b) {$(\kappa,w)$};
\node[draw=black!30!green,below of=layer4a] (layer 5) at ($(layer 4a)!0.5!(layer 4b)$) {$(\kappa,w)$};
\end{scope}

\newcommand{\ofst}{5mm}
\begin{scope}[every node/.style={dashed,draw,minimum width=5cm, minimum height=1cm,rounded corners=1mm}]
\node[DodgerBlue3] at ($(layer 1a)!0.5!(layer 1b)$)(frame 1) {};
\end{scope}
\begin{scope}[every note/.style={minimum width=2cm}]
\node[DodgerBlue3,left=0mm of frame 1,font=\fontsize{8}{8}\selectfont] {\textsf{Predicate Layer}};
\end{scope}
\begin{scope}[every node/.style={dashed,draw,minimum width=5cm, minimum height=1cm,rounded corners=1mm}]
\node[orange] at ($(layer 2a)!0.5!(layer 2b)$)(frame 2) {};
\end{scope}
\begin{scope}[every note/.style={minimum width=2cm}]
\node[orange,left=0mm of frame 2,font=\fontsize{8}{8}\selectfont] {\textsf{Temporal Layer}};
\end{scope}
\begin{scope}[every node/.style={dashed,draw,minimum width=5cm, minimum height=1cm,rounded corners=1mm}]
\node[orange] at ($(layer 3a)!0.5!(layer 3b)$)(frame 3) {};
\end{scope}
\begin{scope}[every note/.style={minimum width=2cm}]
\node[orange,left=0mm of frame 3,font=\fontsize{8}{8}\selectfont] {\textsf{Temporal Layer}};
\end{scope}
\begin{scope}[every node/.style={dashed,draw,minimum width=5cm, minimum height=1cm,rounded corners=1mm}]
\node[black!30!green] at ($(layer 4a)!0.5!(layer 4b)$)(frame 4) {};
\end{scope}
\begin{scope}[every note/.style={minimum width=2cm}]
\node[black!30!green,left=0mm of frame 4,font=\fontsize{8}{8}\selectfont] {\textsf{Boolean Layer}};
\end{scope}
\begin{scope}[every node/.style={dashed,draw,minimum width=5cm, minimum height=1cm,rounded corners=1mm}]
\node[black!30!green] at ($(layer 5)$)(frame 5) {};
\end{scope}
\begin{scope}[every note/.style={minimum width=2cm}]
\node[black!30!green,left=0mm of frame 5,font=\fontsize{8}{8}\selectfont] {\textsf{Boolean Layer}};
\end{scope}

\node[circle,fill=orange!70,minimum size=0.8cm] at ($(layer 1a)!0.5!(layer 1b)+(0,1)$) (s) {$\mathbf{s}$};
\node[circle,fill=orange!70,minimum size=0.8cm, below=5mm of layer 5] (r) {$\mathbf{r}$};

\begin{scope}[-latex]
\draw (s) -- (layer 1a);
\draw (s) -- (layer 1b);
\draw (layer 1a) -- (layer 2a);
\draw (layer 1b) -- (layer 2b);
\draw (layer 2a) -- (layer 3a);
\draw (layer 2b) -- (layer 3b);
\draw (layer 3a) -- (layer 4a);
\draw (layer 3b) -- (layer 4b);
\draw (layer 3a) -- (layer 4b);
\draw (layer 3b) -- (layer 4a);
\draw (layer 4a) -- (layer 5);
\draw (layer 4b) -- (layer 5);
\draw (layer 5) -- (r);
\end{scope}
\end{tikzpicture}
}
\end{subfloat}
\begin{subfloat}[Parameters of TLINet.]{
\begin{tikzpicture}
\begin{scope}[every node/.style={minimum width=2cm,draw}]
\node[draw=DodgerBlue3] (layer 1a) {$(1.0,0.9)$};
\node[draw=DodgerBlue3,right= 3mm of layer 1a] (layer 1b) {$(-1.0,0.7)$};
\node[draw=orange,below=5mm of layer 1a] (layer 2a) {$(-1,0,15)$};
\node[draw=orange] (layer 2b) at (layer 2a-|layer 1b) {$(1,3,7)$};
\node[draw=orange,below=5mm of layer 2a] (layer 3a) {$(1,5,10)$};
\node[draw=orange] (layer 3b) at (layer 3a-|layer 2b) {$(-1,0,10)$};
\node[draw=black!30!green,below=5mm of layer 3a] (layer 4a) {$(1,\bmat{1,0})$};
\node[draw=black!30!green] (layer 4b) at (layer 4a-|layer 3b) {$(-1,\bmat{0,1})$};
\node[draw=black!30!green,below of=layer4a] (layer 5) at ($(layer 4a)!0.5!(layer 4b)$) {$(-1,\bmat{0,1})$};
\end{scope}

\newcommand{\ofst}{5mm}
\begin{scope}[every node/.style={dashed,draw,minimum width=5cm, minimum height=1cm,rounded corners=1mm}]
\node[DodgerBlue3] at ($(layer 1a)!0.5!(layer 1b)$)(frame 1) {};
\end{scope}
\begin{scope}[every node/.style={dashed,draw,minimum width=5cm, minimum height=1cm,rounded corners=1mm}]
\node[orange] at ($(layer 2a)!0.5!(layer 2b)$)(frame 2) {};
\end{scope}
\begin{scope}[every node/.style={dashed,draw,minimum width=5cm, minimum height=1cm,rounded corners=1mm}]
\node[orange] at ($(layer 3a)!0.5!(layer 3b)$)(frame 3) {};
\end{scope}
\begin{scope}[every node/.style={dashed,draw,minimum width=5cm, minimum height=1cm,rounded corners=1mm}]
\node[black!30!green] at ($(layer 4a)!0.5!(layer 4b)$)(frame 4) {};
\end{scope}
\begin{scope}[every node/.style={dashed,draw,minimum width=5cm, minimum height=1cm,rounded corners=1mm}]
\node[black!30!green] at ($(layer 5)$)(frame 5) {};
\end{scope}

\node[circle,fill=orange!70,minimum size=0.8cm] at ($(layer 1a)!0.5!(layer 1b)+(0,1)$) (s) {$\mathbf{s}$};
\node[circle,fill=orange!70,minimum size=0.8cm, below=5mm of layer 5] (r) {$\mathbf{r}$};

\begin{scope}[-latex]
\draw (s) -- (layer 1a);
\draw (s) -- (layer 1b);
\draw (layer 1a) -- (layer 2a);
\draw (layer 1b) -- (layer 2b);
\draw (layer 2a) -- (layer 3a);
\draw (layer 2b) -- (layer 3b);
\draw (layer 3a) -- (layer 4a);
\draw (layer 3b) -- (layer 4b);
\draw (layer 3a) -- (layer 4b);
\draw (layer 3b) -- (layer 4a);
\draw (layer 4a) -- (layer 5);
\draw (layer 4b) -- (layer 5);
\draw (layer 5) -- (r);
\end{scope}
\end{tikzpicture}
}
\end{subfloat}
\begin{subfloat}[TLINet to an STL formula.]{
\begin{tikzpicture}
\begin{scope}[every node/.style={minimum width=2cm,draw}]
\node[draw=DodgerBlue3] (layer 1a) {$x>0.9$};
\node[draw=DodgerBlue3,right= 3mm of layer 1a] (layer 1b) {$x<-0.7$};
\node[draw=orange,below=5mm of layer 1a] (layer 2a) {$\phi_{11}=\always_{[0,15]}$};
\node[draw=orange] (layer 2b) at (layer 2a-|layer 1b) {$\phi_{12}=\event_{[3,7]}$};
\node[draw=orange,below=5mm of layer 2a] (layer 3a) {$\phi_{21}=\event_{[5,10]}$};
\node[draw=orange] (layer 3b) at (layer 3a-|layer 2b) {$\phi_{22}=\always_{[0,10]}$};
\node[draw=black!30!green,below=5mm of layer 3a] (layer 4a) {$\psi_1=\phi_{21}\phi_{11}$};
\node[draw=black!30!green] (layer 4b) at (layer 4a-|layer 3b) {$\psi_2=\phi_{22}\phi_{12}$};
\node[draw=black!30!green,below of=layer4a] (layer 5) at ($(layer 4a)!0.5!(layer 4b)$) {$\psi_2$};
\end{scope}

\newcommand{\ofst}{5mm}
\begin{scope}[every node/.style={dashed,draw, minimum width=5cm, minimum height=1cm,rounded corners=1mm}]
\node[DodgerBlue3] at ($(layer 1a)!0.5!(layer 1b)$)(frame 1) {};
\end{scope}
\begin{scope}[every node/.style={dashed,draw,minimum width=5cm, minimum height=1cm,rounded corners=1mm}]
\node[orange] at ($(layer 2a)!0.5!(layer 2b)$)(frame 2) {};
\end{scope}
\begin{scope}[every node/.style={dashed,draw,minimum width=5cm, minimum height=1cm,rounded corners=1mm}]
\node[orange] at ($(layer 3a)!0.5!(layer 3b)$)(frame 3) {};
\end{scope}
\begin{scope}[every node/.style={dashed,draw,minimum width=5cm, minimum height=1cm,rounded corners=1mm}]
\node[black!30!green] at ($(layer 4a)!0.5!(layer 4b)$)(frame 4) {};
\end{scope}
\begin{scope}[every node/.style={dashed,draw,minimum width=5cm, minimum height=1cm,rounded corners=1mm}]
\node[black!30!green] at ($(layer 5)$)(frame 5) {};
\end{scope}

\node[circle,fill=orange!70,minimum size=0.8cm] at ($(layer 1a)!0.5!(layer 1b)+(0,1)$) (s) {$\mathbf{s}$};
\node[circle,fill=orange!70,minimum size=0.8cm, below=5mm of layer 5] (r) {$\mathbf{r}$};

\begin{scope}[-latex]
\draw (s) -- (layer 1a);
\draw (s) -- (layer 1b);
\draw (layer 1a) -- (layer 2a);
\draw (layer 1b) -- (layer 2b);
\draw (layer 2a) -- (layer 3a);
\draw (layer 2b) -- (layer 3b);
\draw (layer 3a) -- (layer 4a);
\draw (layer 3b) -- (layer 4b);
\draw (layer 3a) -- (layer 4b);
\draw (layer 3b) -- (layer 4a);
\draw (layer 4a) -- (layer 5);
\draw (layer 4b) -- (layer 5);
\draw (layer 5) -- (r);
\end{scope}
\end{tikzpicture} 
}
\end{subfloat}
\caption{An example of TLINet and how it can be transferred to an STL formula from learning parameters.}
\label{fig:NN}
\end{figure*}
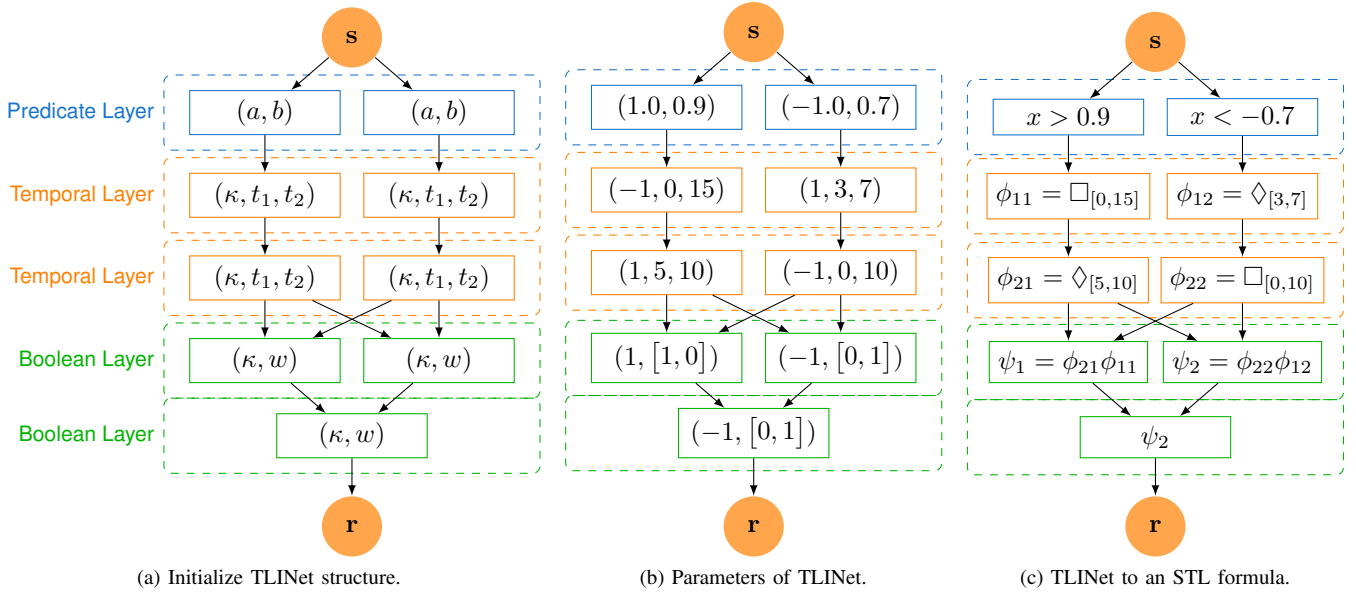

\section{Learning STL formulas with TLINet}
In this section, we introduce in detail how to translate TLINet to an STL formula and the learning process of the TLINet.

\subsection{TLINet as an STL formula}
TLINet, a neural network comprising layers with specific types of operators as defined in Sections \ref{sec:predicate}, \ref{sec:Boolean}, and \ref{sec:temporal}, can be translated into an STL formula by decoding its parameters. Despite the potentially intricate structure of TLINet, the resulting STL formula can be concise. This is because the neural network inherently learns to discard redundant information. Consider a neural network with $5$ layers, as depicted in Figure \ref{fig:NN}. In this example, the first layer consists of $2$ predicate modules, followed by two temporal layers each with $2$ modules, and two Boolean layers with $2$ and $1$ modules, respectively. Such a network can be succinctly translated into an STL formula $\psi = \always_{[0,10]}\event_{[3,7]}(x<-0.7)$.

To yield a succinct STL formula, we introduce a regularizer~$l_s$.
\begin{definition}[Sparsity]
    The sparsity of TLINet is given by $\sum_{j=1}^m (\|\mathbf{w}^b\|_1)_j$ where $\|\cdot\|_1$ represents the $L_1$ norm, $m$ is the number of Boolean modules in TLINet.
\end{definition}
\begin{definition}[Complexity]
    The complexity of an STL formula is the total number of subformulas.
\end{definition}
The sparsity of TLINet plays a crucial role in achieving a compact STL formula. The complexity of the resulting STL formula correlates positively with the sparsity of the neural network. A sparse TLINet, characterized by predominantly zero elements in $\mathbf{w}^b$, yields fewer subformulas, thus reducing formula complexity. Sparsity is encouraged through the regularizer $l_s$, which penalizes non-zero elements in $\mathbf{w}^b$ and contributes to the overall simplicity of TLINet. The regularizer is defined as:
\begin{equation}\label{eq:sparsity reg}
    l_s = \sum_{j=1}^m \sum_{i=1}^{n_j} w^{b,j}_{i},
\end{equation}
where $w^{b,j}_{i}$ denotes the $i^{th}$ element of vector $\mathbf{w}^{b,j}$, $n_j$ is the number of parameters of vector $\mathbf{w}^{b,j}$, and $m$ is the total number of Boolean modules of TLINet.

\subsection{Learning of TLINet}
The learning process in TLINet aims to identify an STL formula that accurately characterizes the observed data, distinguishing between desired (satisfying the STL formula) and undesired (violating it) system behaviors. This task is accomplished through the minimization of specific loss functions tailored for STL inference. Here, we introduce two such loss functions designed for this purpose. Desired and undesired data are typically labeled as 1 and -1, respectively.

\subsubsection{Exponential Loss}
The exponential loss is defined as:
\begin{equation}
    l = e^{-cr},
\end{equation}
where $c$ represents the label of the data, and $r$ denotes the robustness degree of the learned STL formula. This loss function penalizes misclassifications exponentially, making it particularly effective in boosting algorithms.

\subsubsection{Hinge Loss}
The hinge loss, as introduced in \cite{li2024multiclass} is defined as:
\begin{equation}
\begin{aligned}
    l &= \relu(\epsilon-c r) - \gamma \epsilon,
\end{aligned}
\end{equation}
where $c$ is the label of the data, $r$ is the robustness degree of the learned STL formula, $\epsilon$ is the margin, and $\gamma>0$ is a tuning parameter to control the compromise between maximizing the margin and classifying more data correctly~\cite{li2024multiclass}. This loss function explicitly encourages maximizing the margin between classes, potentially leading to better generalization and performance on unseen data.

To incorporate regularization, we augment the loss function with terms introduced in previous sections:
\begin{equation}
    L = l + \lambda_1 l_{s} + \lambda_2 l_{avm} + \lambda_3 l_{kavm}
\end{equation}
where $\lambda_i\in\real{}$, $i\in\{1,2,3\}$; $l_s$ is the sparsity regularizer defined in \eqref{eq:sparsity reg}; $l_{avm}$ is the regularizer for averaged max defined in \eqref{eq:avm reg}; $l_{kavm}$ is the regularizer for averaged minmax function defined in \eqref{eq:kavm reg}. Note that $\lambda_i$ can be zero if the corresponding regularizer is not needed. For example, $\lambda_3=0$ if the averaged minmax function is not used in the neural network. These regularization terms are included to control model complexity and enhance generalization performance.

Given the differentiability of TLINet, state-of-the-art automatic differentiation tools like PyTorch~\cite{paszke2017automatic} are suitable for its implementation. This enables efficient parallelized computation and leverages GPU resources for accelerated training. Additionally, as signals are independent, batch processing techniques can be employed, allowing for efficient scalability, especially for large datasets.

\section{CASE STUDIES}
In this section, we implement several case studies to illustrate the advantages of TLINet. In the first case study, we demonstrate its ability to produce compact STL formulas regardless of the structural complexity. Additionally, we illustrate its computational efficiency compared to other inference methods. In the second case study, we showcase the capability of TLINet to extract various signal features. Lastly, we emphasize its proficiency in capturing complex temporal information within signals through the third case study. We formalize the descriptions of signal features using STL formulas. These STL formulas are inferred through binary classification, where signals with positive labels satisfy the STL formula, while those with negative labels violate it~\cite{bombara2021offline,aasi2022classification}. The neural network TLINet is implemented using PyTorch. The experiments were implemented using a 4.2 GHz 64-core AMD Ryzen CPU and an NVIDIA RTX A5000 GPU.

\subsection{Naval Surveillance Scenario}
In this example, we construct TLINet with different structures and demonstrate their ability to learn concise STL formulas that express signal features in clear mathematical expressions. By comparing TLINet against other methodologies, we highlight its interpretability and efficiency.

We utilize a dataset related to a naval surveillance scenario \cite{aasi2022classification}, as illustrated in Figure \ref{fig:naval}. This dataset contains $1000$ trajectories of length $60$ for each class. Normal behaviors, labeled as $+1$, involve vessels approaching from the sea, navigating through the passage between a peninsula and an island, and proceeding toward the port. Anomalous behaviors, labeled as $-1$, include deviations towards the island or initial adherence to a normal track before returning to the open sea. Our goal is to train the TLINet to learn an STL formula serving as a binary classifier to distinguish desired behaviors from undesirable ones.

\begin{figure}[ht]
    \centering
    \includegraphics[scale=0.4]{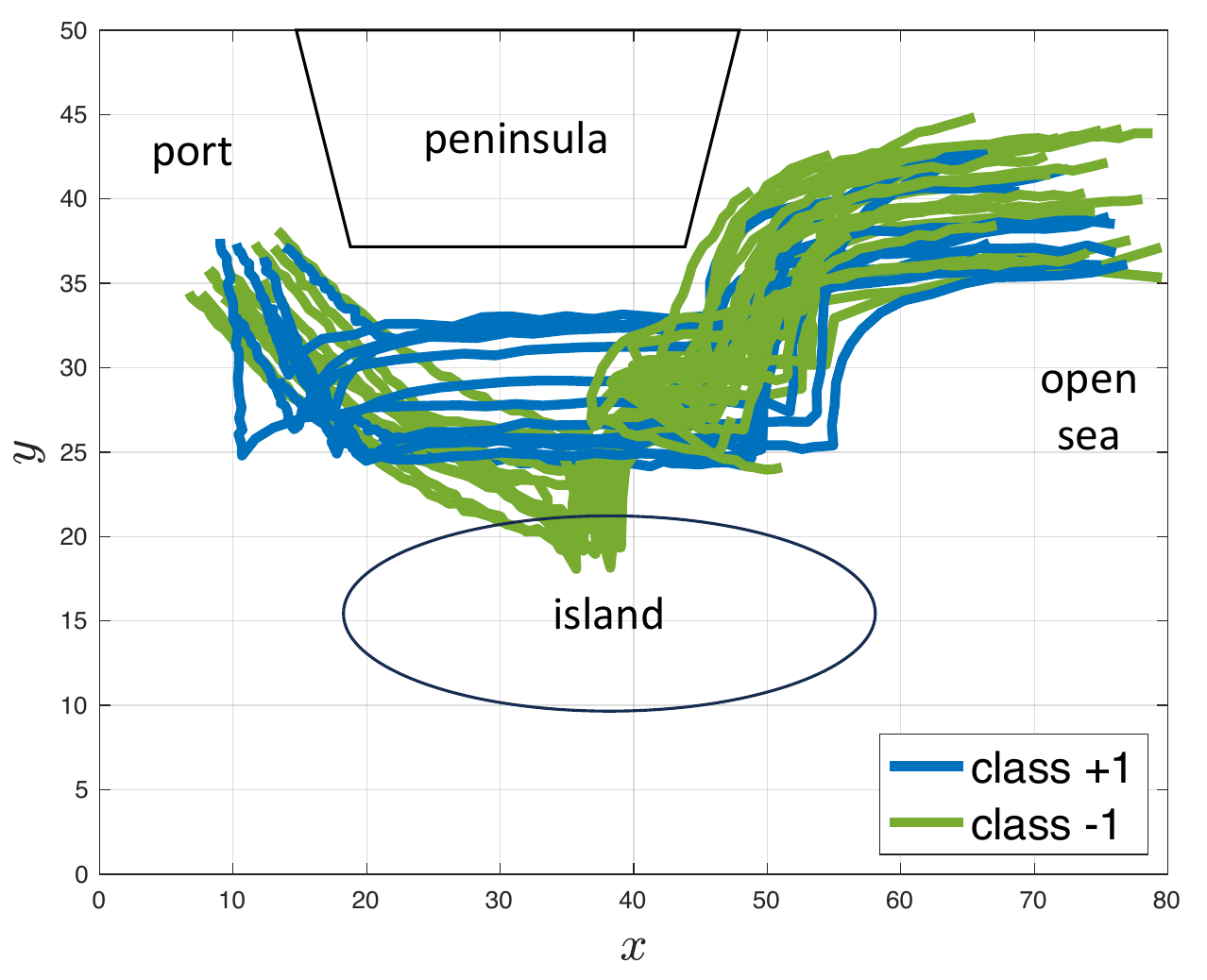}
    \caption{Trajectory examples in a naval surveillance scenario. Trajectories of vessels behaving normally (positive class) are shown in green. The blue trajectories represent the anomalous behaviors (negative class).}
    \label{fig:naval}
\end{figure}

We build three neural networks based on TLINet with increasing structural complexity, denoted as TLINet-1, TLINet-2 and TLINet-3, as illustrated in Figure \ref{fig:naval structure}. We record the training time and the misclassification rate in Table~\ref{table:navalstl}.
All three TLINets achieve an MCR of $0.00$, surpassing the MCR of Boosted Concise Decision Tree (BCDT)~\cite{aasi2022classification}, the decision tree-based framework (DF)~\cite{bombara2021offline} and directed acyclic graph (DAG)~\cite{kong2016temporal}. Moreover, by leveraging the gradient-based methods, our neural network reduces computation time compared to other inference methods, namely BCDT, DF, and DAG. The Long Short-Term Memory (LSTM) network can achieve $0.00$ MCR in a shorter time, but it lacks the ability to describe signal features.

Despite the increasing structural complexity of three TLINets, the resulting STL formulas exhibit comparably concise and uniform forms. The first term can be interpreted as ``the vessel eventually reaches the port" and the second term can be interpreted as ``the vessel always does not reach the island". Thus, these STL formulas are all able to express the signal features. The STL formulas obtained by other classification methods have more complicated structures and involve more terms compared to ours. 
Note that we use the data from \cite{aasi2022classification} which is downsampled from \cite{bombara2021offline} and \cite{kong2016temporal}, accordingly, we adjust the time interval of formulas from \cite{bombara2021offline} and \cite{kong2016temporal} in Table \ref{table:navalstl}  to maintain consistency across the formulas.

\begin{figure}[ht]
\centering
\begin{subfloat}[TLINet-1]{
\centering
\begin{tikzpicture}
\begin{scope}[every node/.style={dashed,draw,minimum width=2cm,,rounded corners=1mm}]
\node[draw=DodgerBlue3,font=\fontsize{9}{9}\selectfont] (layer 1a) {\textsf{P}};
\node[draw=orange,below=3mm of layer 1a,font=\fontsize{9}{9}\selectfont] (layer 2a) {\textsf{T}};
\node[draw=black!30!green,below=3mm of layer 2a,font=\fontsize{9}{9}\selectfont] (layer 3) {\textsf{B}};
\end{scope}

\node[circle,fill=orange!70,minimum size=0.6cm] at ($(layer 1a)+(0,0.9)$) (s) {$\mathbf{s}$};
\node[circle,fill=orange!70,minimum size=0.6cm, below=3mm of layer 3] (r) {$\mathbf{r}$};

\begin{scope}[-latex]
\draw (s) -- (layer 1a);
\draw (layer 1a) -- (layer 2a);
\draw (layer 2a) -- (layer 3);
\draw (layer 3) -- (r);
\end{scope}
\end{tikzpicture}}
\end{subfloat}
~
\begin{subfloat}[TLINet-2]{
\centering
\begin{tikzpicture}
\begin{scope}[every node/.style={dashed,draw,minimum width=2cm,,rounded corners=1mm}]
\node[draw=DodgerBlue3,font=\fontsize{9}{9}\selectfont] (layer 1a) {\textsf{P}};
\node[draw=black!30!green,below=3mm of layer 1a,font=\fontsize{9}{9}\selectfont] (layer 2a) {\textsf{B}};
\node[draw=orange,below=3mm of layer 2a,font=\fontsize{9}{9}\selectfont] (layer 3a) {\textsf{T}};
\node[draw=black!30!green,below=3mm of layer 3a,font=\fontsize{9}{9}\selectfont] (layer 4) {\textsf{B}};
\end{scope}

\node[circle,fill=orange!70,minimum size=0.6cm] at ($(layer 1a)+(0,0.9)$) (s) {$\mathbf{s}$};
\node[circle,fill=orange!70,minimum size=0.6cm, below=3mm of layer 4] (r) {$\mathbf{r}$};

\begin{scope}[-latex]
\draw (s) -- (layer 1a);
\draw (layer 1a) -- (layer 2a);
\draw (layer 2a) -- (layer 3a);
\draw (layer 3a) -- (layer 4);
\draw (layer 4) -- (r);
\end{scope}
\end{tikzpicture}}
\end{subfloat}
~
\begin{subfloat}[TLINet-3]{
\centering
\begin{tikzpicture}
\begin{scope}[every node/.style={dashed,draw,minimum width=2cm,,rounded corners=1mm}]
\node[draw=DodgerBlue3,font=\fontsize{9}{9}\selectfont] (layer 1a) {\textsf{P}};
\node[draw=black!30!green,below=3mm of layer 1a,font=\fontsize{9}{9}\selectfont] (layer 2a) {\textsf{B}};
\node[draw=orange,below=3mm of layer 2a,font=\fontsize{9}{9}\selectfont] (layer 3a) {\textsf{T}};
\node[draw=black!30!green,below=3mm of layer 3a,font=\fontsize{9}{9}\selectfont] (layer 4a) {\textsf{B}};
\node[draw=black!30!green,below=3mm of layer 4a,font=\fontsize{9}{9}\selectfont] (layer 5) {\textsf{B}};
\end{scope}

\node[circle,fill=orange!70,minimum size=0.6cm] at ($(layer 1a)+(0,0.9)$) (s) {$\mathbf{s}$};
\node[circle,fill=orange!70,minimum size=0.6cm, below=3mm of layer 5] (r) {$\mathbf{r}$};

\begin{scope}[-latex]
\draw (s) -- (layer 1a);
\draw (layer 1a) -- (layer 2a);
\draw (layer 2a) -- (layer 3a);
\draw (layer 3a) -- (layer 4a);
\draw (layer 4a) -- (layer 5);
\draw (layer 5) -- (r);
\end{scope}
\end{tikzpicture}}
\end{subfloat}
\caption{The structure of three TLINets. P represents the predicate layer; T represents the temporal layer; B represents the Boolean layer.}
\label{fig:naval structure}
\end{figure}
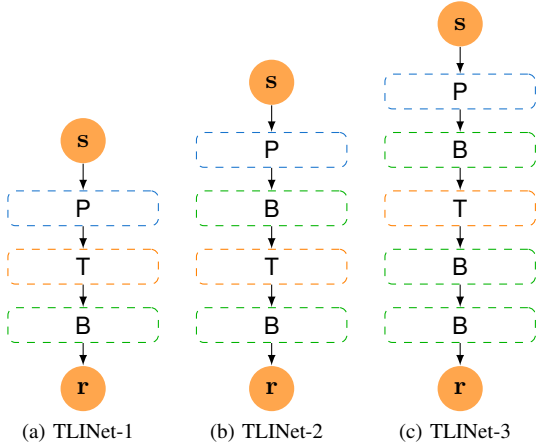
\begin{table*}[ht]
\caption{The results of TLINet and other classification methods.}
\label{table:navalstl}
\centering
{
\begin{tabular}{rccc}
  \toprule
  Method &MCR &Time(s) &STL formula\\
  \midrule
  TLINet-1 &0.0000 &27 &$\event_{[55,60]}(x<25.89) \land \always_{[0,16]}(y>23.77)$\\
  TLINet-2 &0.0000 &103 &$\event_{[58,60]}(x<31.00) \land \always_{[10,16]}(y>24.60)$\\
  TLINet-3 &0.0000 &129 &$ \event_{[49,55]}(x<37.05) \land \always_{[11,14]}(y>23.71)$\\
  LSTM &0.0000 &19 &$/$\\
  BCDT &0.0100 &1996 &$\event_{[28,53]}(x\leq 30.85)\land \always_{[2,26]}((y>21.31) \land (x>11.10))$\\
  DT &0.0195 &140 &$(\neg\event_{[38,53]}(x>20.1) \land \event_{[12,37]}(x>43.2))\lor
  (\event_{[38,53]}(x>20.1) \land \neg\event_{[20,59]}(y>32.2))\lor$\\
  & & &$(\event_{[38,53]}(x>20.1) \land \event_{[20,59]}(y>32.2) \land \always_{[14,60]}(y>30.1))$\\
  DAG &0.0885 &996 &$\event_{[0,33]}(\always_{[18,23]}(y>19.88) \land \always_{[9,30]}(x<34.08))$\\
  \bottomrule
  \end{tabular}
  }
\end{table*}

\subsection{Obstacle Avoidance}\label{sec:obstacle case}
In this example, we show that TLINet is able to capture various characteristics of signals.

Consider the motion planning problem illustrated in Figure \ref{fig:obs}, where the objective is to navigate a robot from the yellow star (position = $(0,0)$) to the target box ($C$) while avoiding the obstacle ($B$). During the data generation process, we consider conventional discrete-time vehicle dynamics and integrate the random sampling of control inputs as reference inputs to create diverse trajectories. 
We define a time-constrained reach-avoid navigation to assess if the robot safely reaches the target within a specified time frame. This serves as the criterion for categorizing the generated trajectories into positive or negative outcomes.
Additionally, the use of a kinodynamic motion planner~\cite{webb2013kinodynamic} allows us to generate more trajectories specifically designed to achieve defined objectives. We generate $1000$ trajectories, each of length $50$, representing scenarios where the objective is achieved, and $1000$ trajectories of the same length representing scenarios where either the target $C$ is not reached or a collision with the obstacle $B$ occurs.

We construct a five-layer TLINet shown in Figure~\ref{fig:obstacle structure} to classify the trajectories. We reach $0.00$ MCR with $122$s training time. The resulting STL formula is
\begin{equation}
\begin{aligned}
    \event_{[46,49]}(7.98<x<11.01 \land 8.00<y<10.55)\\
    \land \always_{[1,49]} (x<3.01\lor x>6.00\lor y<3.14\lor y>4.97).
\end{aligned}
\end{equation}

\begin{figure}[ht]
    \centering
    \includegraphics[scale=0.5]{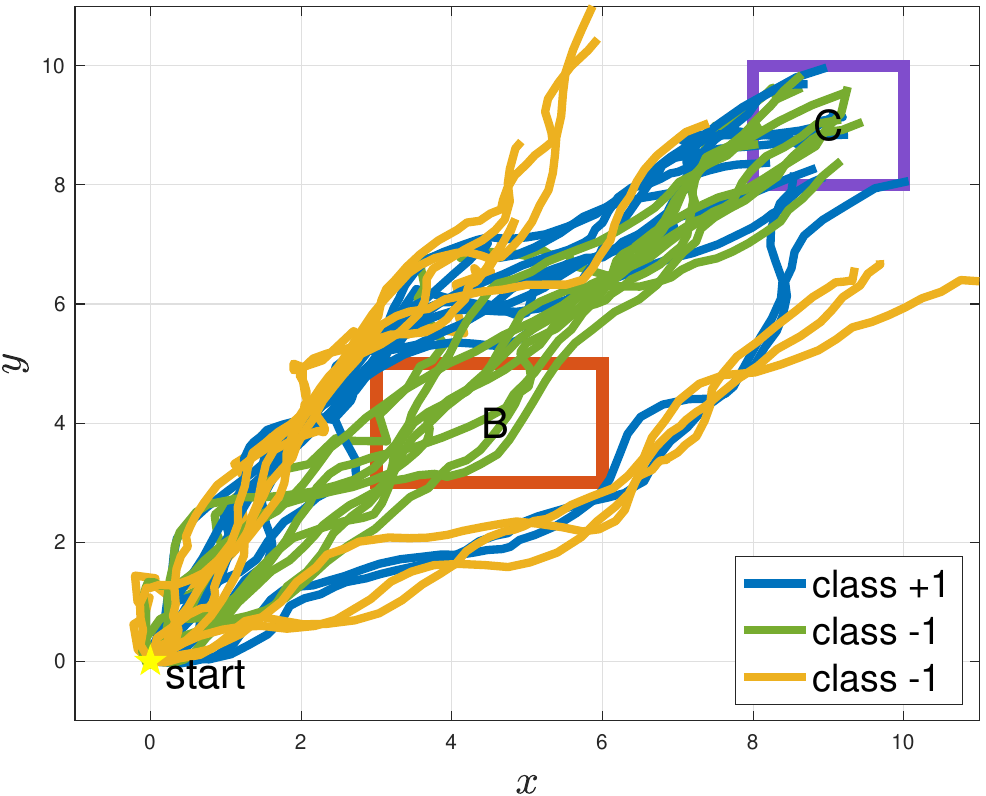}
    \caption{Examples of robot trajectories described in Section \ref{sec:obstacle case}. The blue trajectories (positive class) represent the robots that achieve the objective. The green trajectories (negative class) represent the robots that reach the target $C$ but collide with the obstacle $B$. The yellow trajectories (negative class) represent the robots that avoid the obstacle $B$ but fail to reach target $C$.}
    \label{fig:obs}
\end{figure}
\begin{figure}
\centering
\begin{tikzpicture}
\begin{scope}[every node/.style={dashed,draw,minimum width=5cm, minimum height=0.8cm,rounded corners=1mm}]
\node[draw=DodgerBlue3,font=\fontsize{9}{9}\selectfont] (layer 1a) {\textsf{Predicate Layer}};
\node[draw=black!30!green,below=5mm of layer 1a,font=\fontsize{9}{9}\selectfont] (layer 2a) {\textsf{Boolean Layer}};
\node[draw=orange,below=5mm of layer 2a,font=\fontsize{9}{9}\selectfont] (layer 3a) {\textsf{Temporal Layer}};
\node[draw=black!30!green,below=5mm of layer 3a,font=\fontsize{9}{9}\selectfont] (layer 4) {\textsf{Boolean Layer}};
\end{scope}

\node[circle,fill=orange!70,minimum size=0.8cm] at ($(layer 1a)+(0,1.3)$) (s) {$\mathbf{s}$};
\node[circle,fill=orange!70,minimum size=0.8cm, below=5mm of layer 4] (r) {$\mathbf{r}$};

\begin{scope}[-latex]
\draw (s) -- (layer 1a);
\draw (layer 1a) -- (layer 2a);
\draw (layer 2a) -- (layer 3a);
\draw (layer 3a) -- (layer 4);
\draw (layer 4) -- (r);
\end{scope}
\end{tikzpicture}
\caption{The structure of TLINet for classifying trajectories described in Section \ref{sec:obstacle case}.}
\label{fig:obstacle structure}
\end{figure}
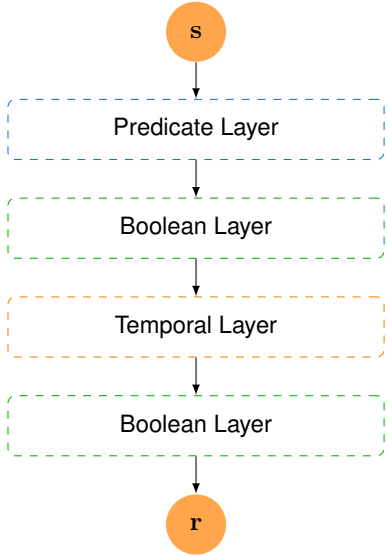

\subsection{Periodic Signal}\label{sec:periodic case}
In this example, we show the capability of TLINet to extract complex temporal information from signals.

Consider the periodic signals illustrated in Figure~\ref{fig:periodic}, the signals are a set of sinusoidal waves with random phase shifts. The dataset contains $1000$ samples with a length of $60$ for each class. The signals with positive labels ($+1$) have a period of $20$ time steps; the signals with negative labels ($-1$) have a period of $40$ time steps. All signals have amplitudes ranging from $1.0$ to $5.0$. Thus, it is crucial to extract temporal information, i.e., the period, for classification.

\begin{figure}[ht]
    \centering
    \includegraphics[scale=0.5]{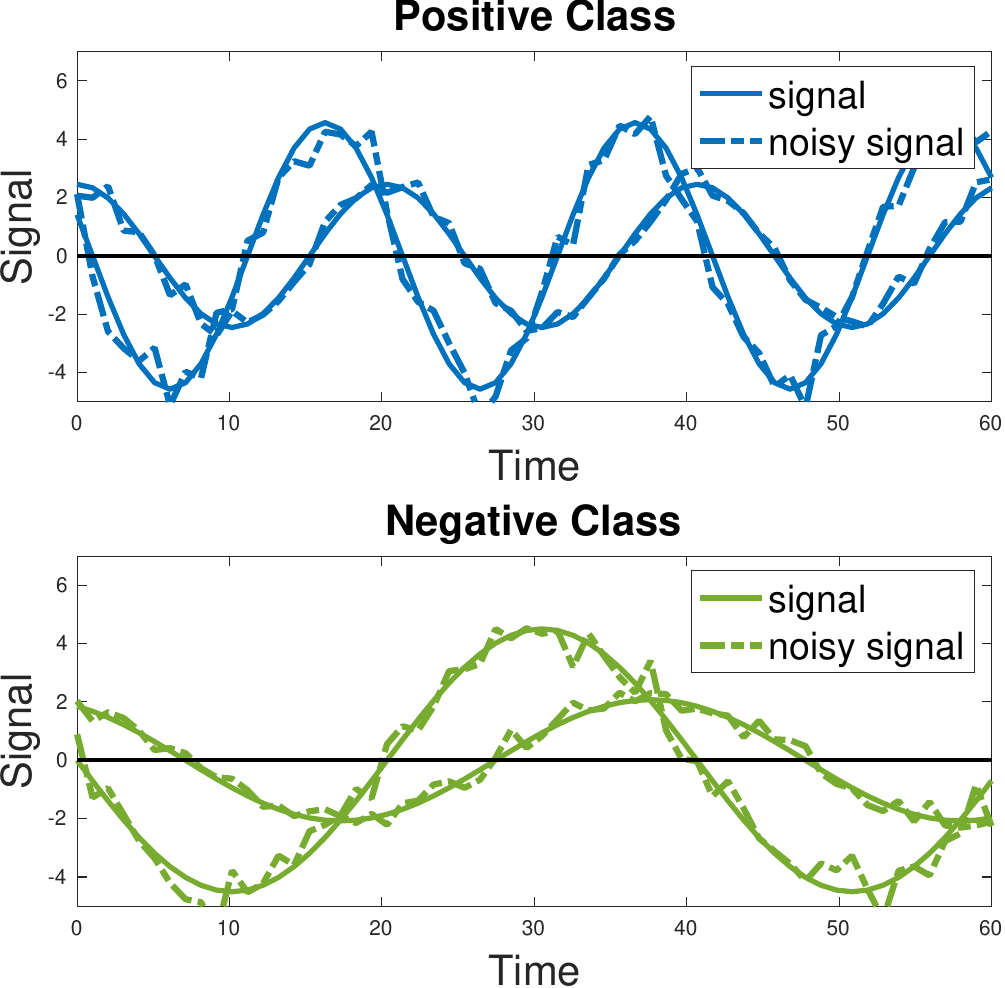}
    \caption{Examples of periodic signals and the noisy signals of two classes described in Section \ref{sec:periodic case}.}
    \label{fig:periodic}
\end{figure}
\begin{figure}
\centering
\begin{tikzpicture}
\begin{scope}[every node/.style={dashed,draw,minimum width=5cm, minimum height=0.8cm,rounded corners=1mm}]
\node[draw=DodgerBlue3,font=\fontsize{9}{9}\selectfont] (layer 1a) {\textsf{Predicate Layer}};
\node[draw=orange,below=5mm of layer 1a,font=\fontsize{9}{9}\selectfont] (layer 2a) {\textsf{Temporal Layer}};
\node[draw=orange,below=5mm of layer 2a,font=\fontsize{9}{9}\selectfont] (layer 3a) {\textsf{Temporal Layer}};
\node[draw=black!30!green,below=5mm of layer 3a,font=\fontsize{9}{9}\selectfont] (layer 4) {\textsf{Boolean Layer}};
\end{scope}

\node[circle,fill=orange!70,minimum size=0.8cm] at ($(layer 1a)+(0,1.3)$) (s) {$\mathbf{s}$};
\node[circle,fill=orange!70,minimum size=0.8cm, below=5mm of layer 4] (r) {$\mathbf{r}$};

\begin{scope}[-latex]
\draw (s) -- (layer 1a);
\draw (layer 1a) -- (layer 2a);
\draw (layer 2a) -- (layer 3a);
\draw (layer 3a) -- (layer 4);
\draw (layer 4) -- (r);
\end{scope}
\end{tikzpicture}
\caption{The structure of TLINet for classifying periodic signals.}
\label{fig:periodic structure}
\end{figure}
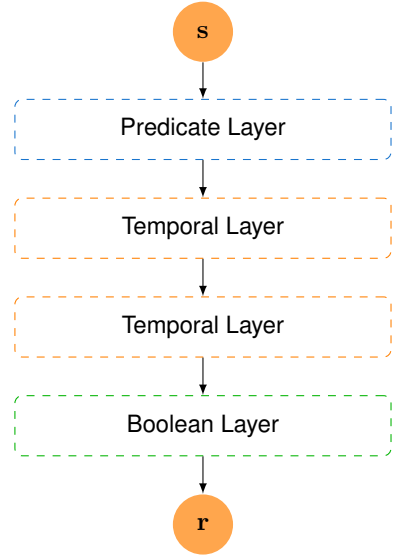
We employ a four-layer neural network, illustrated in Figure~\ref{fig:periodic structure}, featuring two temporal layers to capture the intricate periodic nature of the data. The results from a 5-fold cross-validation, as presented in Table~\ref{table:periodic cross-val}, show a mean $MCR$ of $0.00$ with an average training time of $229$ seconds. The resulting STL formula serves as a description of signals with positive labels. To assess the formula's robustness and generalization against adversarial perturbations, we introduce Gaussian noise to the signals. The noisy signals are also depicted in Figure~\ref{fig:periodic}. Subsequently, we apply the inferred STL formula to classify the noisy signals, and the resulting $MCR$ (noisy) is outlined in Table~\ref{table:periodic cross-val}.

Given the periodicity of the signals, diverse descriptions of features are shown in STL formulas for each cross-validation fold. The structure of STL formulas in Table \ref{table:periodic cross-val} can be summarized as:
\begin{equation}
    \always_{[t_1,t_2]}\event_{[t_3,t_4]} \mu(y).
\end{equation}
The nested structure implies that signals satisfy predicate $\mu(y)$ in a periodic manner, where the inner temporal operator $\event_{[t_3,t_4]} \mu(y)$ is always satisfied at each time step between $[t_1,t_2]$. Since positive-labeled signals consistently cross the zero line within the range from half of the cycle of signals with positive labels to half of the cycle of signals with negative labels, $\mu(y)$ should indicate $y$ pass through the zero line; the length of the time interval $[t_3,t_4]$ needs to be between $10$ to $20$, i.e. $10<t_4-t_3<20$. The time range of the STL formula needs to cover at least a complete cycle of signals with positive labels, ensuring that the temporal pattern is adequately captured. Thus, the condition $(t_2-t_1)+(t_4-t_3)>20$ is necessary to be fulfilled. Consequently, all obtained STL formulas effectively encapsulate the relevant features, facilitating successful classification. For instance, the STL formula from the first fold is $\always_{[0,31]}\event_{[8,20]}(y<-0.08)$, interpreted as ``for all time steps from $0$ to $31$, it is required that eventually, between time steps $8$ and $20$, the variable $y$ is less than $-0.08$." In simpler terms, this signifies that "from time step $8$ to $51$, the variable $y$ periodically drops below $-0.08$ within $12$ time steps."
\begin{table}[ht]
\caption{Cross validation results for classifying periodic signals.}
\label{table:periodic cross-val}
\centering
{
\begin{tabular}{c c p{0.7cm} c c}
  \toprule
  Fold \# &$MCR$ &$MCR$ (noisy) &Time(s) &STL formula\\
  \midrule
  1 &0.000 &0.010 &236 &$\always_{[0,31]}\event_{[8,20]}(y<-0.08)$\\
  2 &0.000 &0.030 &235 &$\always_{[10,45]}\event_{[5,16]}(y<0.03)$\\
  3 &0.000 &0.000 &225 &$\always_{[0,33]}\event_{[6,18]}(y<-0.06)$\\
  4 &0.000 &0.000 &228 &$\always_{[2,35]}\event_{[1,13]}(y>0.07) $\\
  5 &0.000 &0.005 &221 &$\always_{[0,30]}\event_{[4,16]}(y>0.05)$\\
  \midrule
  Mean &0.000 &0.009 &229 \\
  \bottomrule
  \end{tabular}
  }
\end{table}

\section{CONCLUSIONS}
In this paper, we introduce TLINet, a general framework designed for acquiring knowledge of STL formulas from data. TLINet is structured as a computational graph seamlessly incorporating differentiable off-the-shelf computation tools. Our approach is template-free while accommodating nested specifications. Additionally, the approximated robustness of the STL formula is differentiable and crafted with soundness guarantees.
Our experimental results demonstrate TLINet's state-of-the-art classification performance, highlighting its interpretability, compactness, and computational efficiency.
In the future, we will explore extending TLINet into unsupervised learning scenarios, where datasets lack labeling. This extension holds promise for broadening TLINet's real-world applicability.

\section*{APPENDIX}
\subsection{Proof for Proposition \ref{prop:sound}}\label{appendix:soundness}
The sparse softmax function is sound if and only if
\begin{subequations}
\begin{align}
    &\max_{\mathbf{w}}(\mathbf{x})> 0 \iff S(\mathbf{x},\mathbf{w})> 0\label{proofsound1},\\
    &\max_{\mathbf{w}}(\mathbf{x})\leq 0 \iff S(\mathbf{x},\mathbf{w})\leq 0\label{proofsound2}.
\end{align}
\end{subequations}
Note that $q_i$ is always positive and $w_i$ is always non-negative. If $\mathbf{w}$ is a zero vector, from \eqref{vstl:max}, the robustness is meaningless. Thus, $\mathbf{w}$ must be a non-zero vector, then the denominator in \eqref{sparsemax:sub5} is always positive.

First, we give the proof from LHS to RHS. For case \eqref{proofsound1}, let $k=\displaystyle\argmax_i x_iw_i$, from LHS, $x_k w_k>0$, thus $x_k>0$, $w_k=1$, $x_m=x_k$. From \eqref{sparsemax:sub3}, $x_k''=\frac{hx_k w_k}{x_{m}}=h$. Define a function $g(x)=xe^{\beta x}$, the minimum of $g(x)$ is $-\frac{e^{-1}}{\beta}$. 
\begin{equation}
    \begin{aligned}
        \sum_{i=1}^{n} x_i''q_i &= x_k'' q_k+\sum_{i=1, i\neq k}^{n} x_i'' q_i\\
        &=\frac{x_k'' e^{\beta x_k''}+\sum_{i=1, i\neq k}^{n} x_i'' e^{\beta x_i''}}{\sum_{i=1}^{n} e^{\beta x_i''}}\\
        &\geq \frac{h e^{\beta h}-(n-1)\frac{e^{-1}}{\beta}}{\sum_{i=1}^{n} e^{\beta x_i''}}
    \end{aligned}
\end{equation}
Since $\sum_{i=1}^{n} e^{\beta x_i''}>0$, from Proposition \ref{prop:sound}, $h e^{\beta h}-(n-1)\frac{e^{-1}}{\beta}>0$, then $\sum_{i=1}^{n} x_i''q_i> 0$. Since $\sum_{i=1}^{n} x_i''q_i = \sum_{i=1}^{n} \frac{hx_iw_iq_i}{x_m}> 0$ and $\frac{h}{x_m}>0$, $\sum_{i=1}^{n} x_iw_iq_i> 0$. From \eqref{sparsemax:sub5}, $S(\mathbf{x},\mathbf{w})> 0$ is derived.

For case \eqref{proofsound2}, from LHS, $\forall i\in[1,n]$, $x_i w_i \leq 0$. From \eqref{sparsemax:sub5}, $S(\mathbf{x},\mathbf{w})\leq 0$ is derived.

Next, we give the proof from RHS to LHS. For case \eqref{proofsound1}, from RHS, $\exists i\in[1,n]$, $x_i w_i > 0$, thus $\displaystyle\max_{\mathbf{w}}(\mathbf{x})> 0$.

For case \eqref{proofsound2}, from RHS, $\forall i\in[1,n]$, $x_i w_i \leq 0$. Thus, $\displaystyle\max_{\mathbf{w}}(\mathbf{x})\leq 0$ is satisfied.

\bibliographystyle{IEEEtran}
\bibliography{TAC}

\end{document}